\RequirePackage{fix-cm}
\documentclass[smallextended]{svjour3}       
\def\makeheadbox{\relax}

\usepackage{graphicx}
\usepackage{booktabs}
\usepackage{amsmath}
\usepackage{amssymb}
\usepackage[caption=false]{subfig}
\usepackage{natbib}

\newcommand{\trp}{\scriptscriptstyle\top}
\AtBeginDocument{%
  \paperwidth=\dimexpr
    1in + \oddsidemargin
    + \textwidth
    + 1in + \oddsidemargin
  \relax
  \paperheight=\dimexpr
    1in + \topmargin
    + \headheight + \headsep
    + \textheight
    + 1in + \topmargin
  \relax
  \usepackage[pass]{geometry}\relax
}

\begin{document}

\title
{Distributed Adaptive Coverage Control 
of Differential Drive Robotic Sensors 
}

\titlerunning{Distributed Adaptive Coverage of Diff. Drive Robots}        

\author{Rihab Abdul Razak         \and
	Srikant Sukumar  \and
        Hoam Chung 
}


\institute{R. Abdul Razak \at
              IITB-Monash Research Academy \\
              Tel.: +91-7592039470\\
              \email{rihab@sc.iitb.ac.in}.           
	   \and
	   S. Srikant \at
	      Dept. of Systems and Control Engg., Indian Institute of Technology Bombay.
           \and
           H. Chung \at
              Dept. of Mechanical and Aerospace Engg., Monash University.
}

\date{}

\maketitle

\begin{abstract}
This paper is concerned with the deployment of multiple mobile
robots in order to autonomously cover a region $\mathcal{Q}$.
The region to be covered is described using a density function
which may not be apriori known.
In this paper, we pose the coverage problem as an optimization
problem over some space of functions on $\mathcal{Q}$. In
particular, we look at $\mathcal{L}^2$-distance based coverage
algorithm and derive adaptive control laws for the same. We also
propose a modified adaptive control law incorporating consensus
for better parameter convergence. We implement the algorithms on
real differential drive robots with both simulated density
function as well as density function implemented using light
sources. We also compare the $\mathcal{L}^2$-distance based
method with the locational optimization method using experiments. 
\end{abstract}

\section{Introduction}
\label{sec:intro}
Cooperative control problems involving multi-agent systems have
been widely studied in the literature. We consider multiple
autonomous agents that work together to achieve an objective.
Objectives include rendezvous where the agents try to converge
to a common state, formation control where the agents try to
maintain a given spatial formation and coverage control where
the agents are deployed to cover a given region of interest.
See
\citet{jadbabaie2003coordination, murray2007recent,
olfati2007consensus, tanner2007flocking,
cortes2004coverage,
DistCtrlRobotNetw, Song2011, Song2013}.
These cooperative control algorithms find applications in
surveillance, patrolling, environmental monitoring and sensing etc.
\par The agents communicate with each other using some
communication topology which is described in terms
of a graph where the nodes correspond to the agents and two nodes
are connected if the two corresponding agents can communicate
with each other. In most cases, communication graphs correspond
to proximity graphs meaning that two agents communicate if they
are close to each other. This also motivates the use of
decentralized or distributed control strategies for efficient
solution of multiagent problems where the control laws of
individual agents are determined by the information exchange
with their neighbouring agents See for example
\citet{murray2007recent, olfati2007consensus, DistCtrlRobotNetw}.
\par In this paper, we consider the problem of optimally covering
a given region using multiple agents to sense a phenomenon/event of
interest. The event of interest is described by a
\textit{density function} over the region. The density function
can be thought of as giving the distribution of the intensity
of the phenomenon to be sensed. For example, in case of mobile
agents deployed to sense nuclear radiation over a region, the
density function could be the intensity of radiation over the
region. In this case, we would like the mobile agents which are
deployed starting at some initial position to converge to some
optimal configuration for sensing purpose.
\par The coverage problem in the locational optimization
framework was investigated in in~\citet{cortes2004coverage}, where
$n$ agents are deployed to cover a convex region
$\mathcal{Q}\subset\mathbb{R}^q$. The problem was solved for
agents with single integrator dynamics and known density
function. In \citet{SchwagerEtalICRA07} and \citet{Schwager2009},
the authors extend the algorithm of \citet{cortes2004coverage}
using adaptive control for the case where the density function
is not fully known. The density function is assumed to be
linearly parameterized in terms of a vector of unknown constant
parameters. They also propose a consensus term in the adaptation
law which improves parameter convergence. In
\citet{cortes2005geometric}, the authors discuss spatial
optimization problems closely related to coverage problems
using gradient descent methods. In
\citet{hexsel2011,guruprasad2013,bopardikar2018arXiv},
the authors talk about different versions of the
coverage problem based on locational optimization by using
different assumptions on the agent sensing capabilities. The
adaptive coverage algorithms have been extended to nonholonomic
robots in \citet{Luna2013,rihab2018a}.
In many problems, it would be also be beneficial to estimate
the parameters of the density function correctly along with
the coverage task. In this case, the issue of parameter
estimation becomes important. The parameters converge to
true values provided an integral condition over the trajectories
of the agents are satisfied which may not always be achieved.
There are not many works in the literature which focus on the
issue of parameter estimation in the context of coverage control.
\par In this work, we pose the coverage problem in a general
framework. The agents located at different positions in the
domain can be thought of as defining an \textit{agent density
function}. The coverage problem can then be posed as an
optimization problem which seeks to minimize the \textit{distance}
between the original density function and the agent density
function for some appropriately defined distance. This approach
is more general in the sense that the locational optimization
problem can be viewed as a special case of this formulation.
We in addition look at the
$\mathcal{L}^2$-metric for achieving coverage and derive adaptive
control laws for differential drive robots. We also present a
slightly modified adaptive control law using consensus over
directed sub-graphs of the delaunay graph for improving parameter
convergence. The algorithms described are tested on actual
differential drive robots and a comparison is given between the
$\mathcal{L}^2$-distance based approach and the locational
optimization approach in \citet{cortes2004coverage,rihab2018a}.
We are particularly interested in studying how the two approaches
handle the problem of parameter estimation. A part of the current
work has been presented at \citet{rihab2018b}.
\par In section \ref{sec:problem}, we formulate and discuss the
coverage problem. In section \ref{sec:l2coverage}, we discuss a
new objective function based on $\mathcal{L}^2$-distance for
achieving coverage. In section \ref{sec:control} we derive
control and adaptation laws for differential drive robots to
converge to near optimal configuration. We also discuss some
improvements to the parameter adaptation law which gives better
parameter convergence. In section \ref{sec:hardware} we discuss
experimental results on differential drive robots and also give
a comparison of the $\mathcal{L}^2$ method with the locational
optimization. Finally we conclude the paper with section
\ref{sec:conclusions}.

\section{Problem Formulation}
\label{sec:problem}
We consider a bounded convex region
$\mathcal{Q} \subset \mathbb{R}^n$ where $N$ agents are deployed
so as to spread and distribute themselves in an optimal manner
with respect to a density function
$\phi: \mathcal{Q} \to \mathbb{R}_+$ where $\mathbb{R}_+$ is the
set of positive real numbers. The density function $\phi(.)$
represents the event of interest with respect to which coverage
is to be obtained (also called $\phi(.)$ the target density).
Intuitively we want more robots to be concentrated over regions
having higher values of the density function. The positions of
the agents are denoted by $p_i$ for $i=1,2,\dots,N$ and the set
of all agent positions is denoted by $P$. For a given
configuration of agents $P = \{p_i\}_{i=1}^N$, we define the
voronoi partition of $\mathcal{Q}$ as the set
$\{\mathcal{V}_i\}_{i=1}^N$ where
\begin{equation}
 \centering
 \mathcal{V}_i = \{ q \in \mathcal{Q} \, \mid \, \|q-p_i\| \leq
 \|q-p_j\|, \,\, \forall j \neq i \}.
 \label{eqn:voronoi}
\end{equation}
The voronoi cell $\mathcal{V}_i$ is the set of all points closest
to agent $i$ compared to all the other agents. In
\citet{cortes2004coverage}, the optimal coverage configuration is
described as the optimum of the locational optimization cost
\begin{equation}
 \centering
 \mathcal{H}(P) = \sum_{i=1}^N \int\limits_{\mathcal{V}_i}
 \|p_i - q\|^2 \phi(q) dq
\end{equation}
In this paper, we formulate the coverage problem in a more
general framework, as an optimization problem over a space of
functions over $\mathcal{Q}$.

\subsection{Distance Function based Approach}
\label{sec:distancebasedapproach}

Each agent is assumed to have a sensing capability which
decreases with the distance from the agent location. We quantify
this sensing capability of each agent as the
\textit{sensing function} denoted by $f_i(p_i,q)$. The sensing
function describes the capability of the agent at position $p_i$
to sense event at point $q$. In this work we assume isotropic
sensors whose sensing is independent of direction. We can thus
represent the sensing function as $f_i(\|p_i-q\|)$. We require
$f_i : \mathbb{R}_+ \to \mathbb{R}_+$ to be an appropriate
decreasing function of its argument. An illustration for the
one-dimensional case is shown in figure \ref{fig:problem}.
\begin{figure}[h]
 \centering
 \includegraphics
 [width=0.88\textwidth,height=0.22\textheight]
 {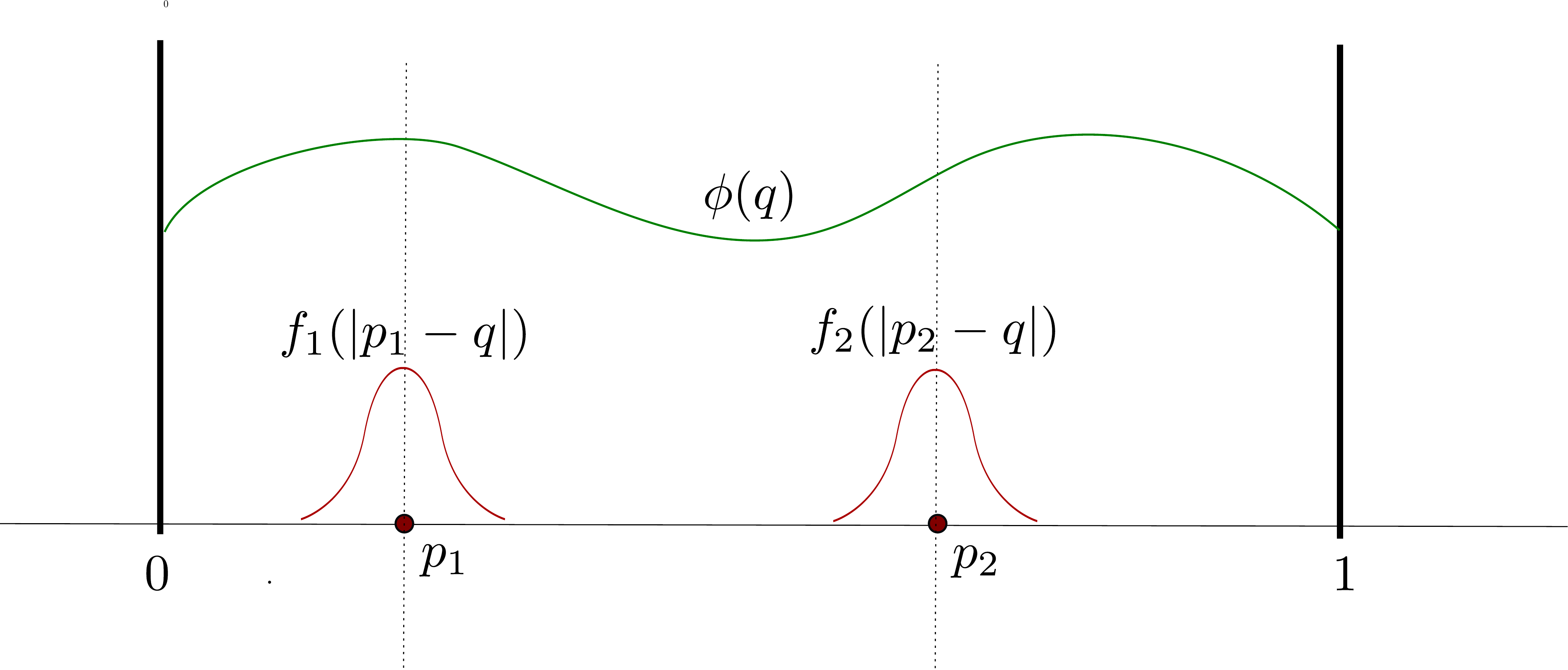}
 \caption{Two agents covering the interval $\mathcal{Q} = [0,1]$}
 \label{fig:problem}
\end{figure}
Given $N$ agents each with its sensing functions, we define an
\textit{aggregate agent density function} defined by
\[
 f_P(q) :=
 \alpha(f_1(\|p_1-q\|),f_2(\|p_2-q\|),\dots,f_N(\|p_N-q\|))
\]
where $\alpha : \mathbb{R}_+^N \to \mathbb{R}_+$ is an aggregation
function which is to be chosen appropriately. The aggregate agent
density gives a measure of the quality of sensing of the region
$\mathcal{Q}$ using all the $N$ agents. The general coverage
problem can then be posed as the following optimization problem:
\begin{equation}
  \min_{P=\{p_1,p_2,\dots,p_N\}} d(\phi,f_P)
  \label{eqn:problem}
\end{equation}
where $d(.,.)$ is an appropriate measure of divergence or distance
between $\phi$ and $f_P$. The agents should move to a
configuration such that the \textit{distance} between the target
density $\phi$ and the aggregate agent density $f_P$ is minimized.
The choice of $d(.,.)$ will generally be determined by the class
of functions we work with which depends on the nature of the
coverage density function $\phi(.)$ as well as the nature of
sensors in the mobile agents employed. The optimal locations of
the agents is then given by
\begin{equation}
 P^* = \{p_1^*,p_2^*,\dots,p_N^*\} =
 \mbox{arg} \min_{P} d(\phi,f_P)
\end{equation}

\subsection{Choice of the agent sensing functions and
aggregate density}
\label{sec:choiceofdensity}
The agent sensing functions $f_i(\|p_i-q\|)$ as mentioned above
is required to be a non-increasing function of its argument, and
depends on the nature of sensors. Some examples which we can be
used are given below:
\begin{enumerate}
 \item Gaussian function
 \begin{equation}
  f_i(x) =
  A_i \exp\left( -\frac{x^2}{\sigma_i^2} \right)
  \label{eqn:agentdensity_gaussian}
 \end{equation}
 for some constants $A_i > 0$, $\sigma_i > 0$.
 \item Constant sensing function
 \begin{equation}
  f_i(x) = \mathcal{\chi}_{\{x \leq r_i\}}
  \label{eqn:agentdensity_constant}
 \end{equation}
 where $\mathcal{\chi}_S$ is the characteristic function of the
 set $S$ and $r_i > 0$ is some constant.
 \item Quartic function
 \begin{equation}
  f_i(x) =
  \left\{ \begin{array}{ll}
  \frac{M}{r_i^4} \left( x^2-r_i^2 \right)^2
  & \mbox{ if } x < r_i \\
  0 & \mbox{ if } x \geq r_i
  \end{array} \right.
  \label{eqn:agentdensity_quartic}
 \end{equation}
 for some constants $M > 0, r_i > 0$.
 \item Bump function
 \begin{equation}
  f_i(x) =
  \left\{ \begin{array}{ll}
  A_i \exp\left\{-\frac{1}{1-\left(\frac{x}{r_i}\right)^2}\right\}
  & \mbox{ if } x < r_i \\
  0 & \mbox{ if } x \geq r_i
  \end{array} \right.
  \label{eqn:agentdensity_bump}
 \end{equation}
 for some constant $A_i > 0$, $r_i > 0$.
\end{enumerate}
The constant sensing function, the bump function and the quartic
function (unlike the Gaussian function) allow us to model sensors
with a finite sensing range since they take zero value outside a
finite region of radius $r_i$ around the agent position.
Accordingly, we consider these functions as examples of
\textit{limited range} sensing functions, and the Gaussian
function as an example of a \textit{full range} sensing function.
The choice of aggregate agent density $f_P(\cdot)$ could be made in
different ways. We give here two natural choices:
\begin{enumerate}
 \item Average/sum function
 \begin{equation}
  f_P(q) = \gamma \sum_{i=1}^{N} f_i(\|p_i - q\|)
  \label{eqn:avgaggregatefcn}
 \end{equation}
 \item Max function
 \begin{equation}
  f_P(q) = \gamma \max_{i} f_i(\|p_i - q\|)
  \label{eqn:maxaggregatefcn}
 \end{equation}
\end{enumerate}
where $\gamma$ is some positive constant. The max function
defines the aggregate density at a point as the value of the
agent sensing function with the maximum value at the point.
If the agents are identical, then the aggregate density at a
point in this case is the agent density corresponding to the
agent that is closest to the point. We will see that this
naturally leads to a partition of the region $\mathcal{Q}$ and
allows using distributed control schemes.

\subsection{Locational Optimization}
\label{sec:1-2}
Now, we consider the locational optimization problem defined by
the cost function [see \citet{cortes2004coverage}]
\begin{align}
 \mathcal{H}(p_1,\dots,p_N) &=
 \int_Q \min_i \|p_i-q\|^2 \phi(q) dq \\
 &= \sum_{i=1}^N \int_{\mathcal{V}_i} \|p_i-q\|^2 \phi(q) dq
\end{align}
where the domain $\mathcal{Q}$ is assumed to be compact and
convex, $p_i$ are the locations of the agents and
\{$\mathcal{V}_i$\} is the voronoi partition of $\mathcal{Q}$
corresponding to $\{p_i\}$. The minimizer of the function
$\mathcal{H}$ correspond to the centroidal voronoi configuration.
These are the points $p_i^*$ such that $p_i^* = C_{\mathcal{V}_i}$
where $C_{\mathcal{V}_i} = (\int_{\mathcal{V}_i} q \phi(q) dq)/
(\int_{\mathcal{V}_i} \phi(q) dq)$.
\par We wish to pose the locational optimization problem in the
general framework discussed above (equation (\ref{eqn:problem})).
In order to do that we consider the relative entropy or
Kullback-Leibler divergence of two integrable functions
$f:\mathcal{Q} \to \mathbb{R}_{+}$ and
$g:\mathcal{Q} \to \mathbb{R}_{+}$ where $g(q)>0$, defined as
\begin{equation}
 d_{KL}(f,g) = \int_Q f(q) \log(\frac{f(q)}{g(q)}) dq
\end{equation}
$d_{KL}$ is not a metric (since it is not symmetric) but it is
a measure of divergence (distance) of $g$ with respect to $f$
and $d_{KL}(f,f)=0$
[see \citet{kullback1951,numericalrecipes2007}].
{
If we interpret the density function $\phi$ as defining the
probability density of events occurring in the domain
$\mathcal{Q}$, and agent sensing functions as defining the
probability of detecting an event, then minimizing
$d_{KL} (\phi,f_P)$ would mean that we are choosing those agent
positions which results in maximum likelihood of detection of
these events.
}
\begin{proposition}
Assume that
\begin{itemize}
 \item The target density $\phi(.)$ is integrable.
 \item The sensing functions of the agents are gaussian
 (equation (\ref{eqn:agentdensity_gaussian})) with $A_i = 1$ and
 $\sigma_i = 1$.
 \item The aggregate density function of the agents is given by
 the max function (\ref{eqn:maxaggregatefcn}) with $\gamma = 1$.
\end{itemize}
Under the above assumptions,
\begin{equation*}
 \underset{P}{\arg\min} \,\,\, d_{KL}(\phi,f_P) =
 \underset{P}{\arg\min} \,\,\, \mathcal{H}(P)
\end{equation*}
i.e. minimizing $d_{KL}(\phi,f_P)$ is equivalent to minimizing
$\mathcal{H}(P)$.
\end{proposition}
\begin{proof}
For each agent $i$
\[
 f_i(\|p_i-q\|) = \exp(-\|p_i-q\|^2)
\]
and the aggregate density is given by
\[
 f_P(q) = \max_i \, \exp(-\|p_i-q\|^2)
\]
Then,
\begin{align*}
 d_{KL}(\phi,f_P) &=
 \int_Q \phi(q) \log\frac{\phi(q)}{f_P(q)} dq \\
 &= \int_Q \phi(q) \log \phi(q) dq -
 \int_Q \phi(q) \log f_P(q) dq
\end{align*}
Since the first term is independent of $P$, we have
\begin{align*}
 \min_{P} d_{KL}(\phi,f_P) &=
 \min_{P} -\int_Q \phi(q) \log f_P(q) dq  \\
 &=\min_{P} -\int_Q \phi(q) \log \max_i f_i(\|p_i-q\|) dq \\
 &=\min_{P} -\int_Q \phi(q) \max_i \log f_i(\|p_i-q\|) dq \\
 &=\min_{P} \int_Q \phi(q) \min_i \|p_i-q\|^2 dq \\
 &=\min_{P} \mathcal{H}(P)
\end{align*}
\end{proof}
It can be shown [\citet{cortes2004coverage}] that for agents with
single integrator dynamics $\dot{p}_i = u_i$, the gradient
control law given by
\[
 u_i = -\frac{\partial\mathcal{H}}{\partial p_i}
 = -kM_{\mathcal{V}_i}(p_i - C_{\mathcal{V}_i})
\]
allows the agents to coverge to the centroidal voronoi
configuration.

\section{$\mathcal{L}^2$-distance based coverage}
\label{sec:l2coverage}
We assume that the target density $\phi$ belong to
$\mathcal{L}^2(\mathcal{Q})$. We further assume that $\phi$ is
lower bounded, i.e., $\phi(q) \geq \beta$ for some constant
$\beta > 0$. Assuming that the aggregate density of the agents
also belong to $\mathcal{L}^2(\mathcal{Q})$, we can define the
following cost function
\begin{equation}
 d_2(\phi,f_P) = \int_Q | \phi(q) - f_P(q) |^2 dq
 \label{eqn:d2}
\end{equation}
We now investigate the multiagent coverage problem using the
above cost function.
Any minima of the cost function $d_2(.,.)$ with respect to $P$
satisifies the condition $\frac{\partial d_2}{\partial p_i} = 0$
for every $i$.
\noindent Let $\mathcal{Q} \subset \mathbb{R}^n$ be a convex and
bounded region with $N$ agents and let
\begin{equation*}
 \begin{aligned}
  f_P(q) &= \gamma \max_i f(\| p_i-q \|)
 \end{aligned}
\end{equation*}
where the agent sensing function $f(.)$ is assumed to be the
same for all the agents.
Then
\begin{equation}
 \begin{aligned}
  d_2(\phi,f_P) &= \int_{\mathcal{Q}}
  \left| \phi(q) - f_P(q) \right|^2 dq \\
  &= \int_{\mathcal{Q}} \left| \phi(q) -
  \gamma \left[\max_i f(\|p_i - q\|)\right] \right|^2 dq
 \end{aligned}
\end{equation}
Since $f(.)$ is a decreasing function of its argument, we can
write the above as
\begin{equation}
 d_2(\phi,f_P) = \sum_{i=1}^N \int_{\mathcal{V}_i}
 \left| \phi(q) - \gamma . f(\|p_i-q\|) \right|^2 dq
 \label{eqn:costl2}
\end{equation}
where $\{\mathcal{V}_i\}_{i=1}^N$ is the voronoi partition
defined by equation (\ref{eqn:voronoi}).
\begin{lemma}
 The gradient of the $\mathcal{L}^2$ cost function
 (\ref{eqn:costl2}) with respect to $p_i$ is given by
 \begin{equation}
  \frac{\partial d_2}{\partial p_i} =
  \int_{\mathcal{V}_i} \frac{\partial}{\partial p_i}
  \left| \phi(q) - \gamma.f(\|p_i-q\|) \right|^2 dq
  \label{eqn:derivd2}
 \end{equation}
\end{lemma}

\begin{proof}
Computing the gradient with respect to $p_i$,
\begin{equation}
 \begin{aligned}
  \frac{\partial d_2}{\partial p_i} &=
  \frac{\partial}{\partial p_i} \int_{\mathcal{V}_i}
  \left| \phi(q) - \gamma . f(\|p_i-q\|) \right|^2 dq \\
  & \,\, + \sum_{j \in \mathcal{N}_i}
  \frac{\partial}{\partial p_i}
  \int_{\mathcal{V}_j} \left| \phi(q) - \gamma .
  f(\|p_j-q\|) \right|^2 dq
 \end{aligned}
\end{equation}
It should be noted that in the expression above the regions of
integration $\mathcal{V}_i$ and $\mathcal{V}_j$ are themselves
functions of $p_i$. Thus we have
[see \citet{Cortes2002,flanders73}]
\begin{equation}
 \begin{aligned}
  \frac{\partial d_2}{\partial p_i} &=
  \int_{\mathcal{V}_i} \frac{\partial}{\partial p_i}
  \left| \phi(q) - \gamma.f(\|p_i-q\|) \right|^2 dq \\
  & \,\, + \int_{\partial \mathcal{V}_i}
  \left| \phi(q) - \gamma.f(\|p_i-q\|) \right|^2
  n_{iq}^{\trp} \frac{\partial q}{\partial p_i} dq \\
  & \,\, + \sum_{j \in \mathcal{N}_i}
  \int_{c_{ij}} \left| \phi(q) - \gamma.f(\|p_j-q\|)
  \right|^2 n_{jq}^{\trp} \frac{\partial q}{\partial p_i} dq
 \end{aligned}
\end{equation}
where $\partial \mathcal{V}_i$ is the boundary of the voronoi
region $\mathcal{V}_i$ , $\mathcal{N}_i$ is the set of voronoi
neighbours of agent $i$, $n_{iq}$ is the unit normal at point
$q$ pointing outward from the region $\mathcal{V}_i$ and
$c_{ij}$ is the boundary segment of voronoi region
$\mathcal{V}_i$ that is shared with the voronoi region of agent
$j$ ($\mathcal{V}_j$) [see also \citet{Cortes2002}].
The boundary $\partial \mathcal{V}_i$ consists of segments
$c_{ij}$ and possibly parts of $\partial \mathcal{Q}$
(the boundary of $\mathcal{Q}$).
The integrand of the second term is zero over
$\partial \mathcal{Q}$. Thus the second and third terms in the
above expression cancel each other (since the outward normals
$n_{iq}$ and $n_{jq}$ in the two terms point opposite to each
other) and the proof is complete.
\end{proof}

As a simple example, consider $\mathcal{Q}=[0,1]$
with one agent with $f(|p-q|) = \exp\{-|p-q|^2\} =: f_P(q)$.
Assume the density function is given by
$\phi(q) = \exp\{-|c-q|^2\}$ for some $c \in (0,1)$. Minimizing
the $\mathcal{L}^2$ norm implies
$\frac{\partial d_2}{\partial p} = 0$
which implies from (\ref{eqn:derivd2}) that $p=c$. \\

\subsection{Gaussian sensing function}
\par Now consider the Gaussian agent density function,
$f(\|p_i-q\|)=\exp\{ -\frac{\|p_i - q\|^2}{\sigma^2}
\},\,\, i=1,2,\dots,N$.
In this case the gradient becomes
\begin{equation*}
 \begin{aligned}
  \frac{\partial d_2}{\partial p_i} &=
  \int_{\mathcal{V}_i} \frac{\partial}{\partial p_i}
  \left| \phi(q) - \gamma.f(\|p_i-q\|) \right|^2 dq \\
  &= \frac{4 \gamma}{\sigma^2} \int_{\mathcal{V}_i}
  \left( \phi(q) - \gamma.f(\|p_i-q\|) \right)
  \left( p_i-q \right)
  e^{-\frac{\|p_i-q\|^2}{\sigma^2}} dq \\
  &= \frac{4 \gamma}{\sigma^2} \left\{ p_i \int_{\mathcal{V}_i}
  \left[ \phi(q) - \gamma.e^{-\frac{\|p_i-q\|^2}{\sigma^2}}
  \right]
  e^{-\frac{\|p_i-q\|^2}{\sigma^2}} dq \right. \\
  & \qquad - \left. \int_{\mathcal{V}_i} q
  \left[ \phi(q) - \gamma.e^{-\frac{\|p_i-q\|^2}{\sigma^2}}
  \right]
  e^{-\frac{\|p_i-q\|^2}{\sigma^2}} dq \right\} \\
  &= \frac{4 \gamma}{\sigma^2} \left\{ p_i \int_{\mathcal{V}_i}
  \lambda_i(q) dq - \int_{\mathcal{V}_i} q \lambda_i(q) dq
  \right\}
 \end{aligned}
\end{equation*}
where
\begin{equation}
 \lambda_i(q):=
  e^{-\frac{\|p_i-q\|^2}{\sigma^2}}
  \left[ \phi(q) -
  \gamma \cdot e^{-\frac{\|p_i-q\|^2}{\sigma^2}}
  \right].
 \label{eqn:lambdai}
\end{equation}
Thus we can write
\begin{equation}
 \frac{\partial d_2}{\partial p_i} =
 4 \gamma M_{\mathcal{V}_i}^{\lambda}
 (p_i - C_{\mathcal{V}_i}^{\lambda})
 \label{eqn:derivd2_gaussian}
\end{equation}
where we define
\begin{equation}
 \mathcal{C}_{\mathcal{V}_i}^{\lambda} :=
 \frac{L_{\mathcal{V}_i}^{\lambda}}{M_{\mathcal{V}_i}^{\lambda}}
 = \frac{\int_{\mathcal{V}_i} q \lambda_i(q) dq}
 {\int_{\mathcal{V}_i} \lambda_i(q) dq}
 \label{eqn:cvilambdadef}
\end{equation}
In order to make sure that
$C_{\mathcal{V}_i}^{\lambda}$ is well-defined, we need
the following:

\begin{lemma}
 For $0 < \gamma \leq \beta$
 (recall that $\beta$ is the lower bound on $\phi$),
 $\lambda_i(q) \geq 0$ for $q \in \mathcal{V}_i$ and
 $i=1,2,\dots,N$.
 \label{lemma:gammalemma}
\end{lemma}
\begin{proof}
The second term in equation (\ref{eqn:lambdai}) an exponential
and is always positive. The first term (in the square brackets)
is non-negative if $0 < \gamma \leq \beta$. Thus proved.
\end{proof}
\begin{remark}
From the definition of
$\mathcal{C}_{\mathcal{V}_i}^{\lambda}$ and lemma
\ref{lemma:gammalemma}, we conclude that
$\mathcal{C}_{\mathcal{V}_i}^{\lambda} \in \mathcal{V}_i$ for
$0 < \gamma \leq \beta$, i.e.
$\mathcal{C}_{\mathcal{V}_i}^{\lambda}$ is contained in the
convex set $\mathcal{V}_i$.
\end{remark}
\begin{remark}
Choosing $\gamma$ as per the above lemma guarantees that
the aggregate density function of the agents $f_P(.)$ is scaled
below the lower bound of $\phi(.)$ function. This might be
detrimental to matching $f_P$ and $\phi$ if the lower bound
$\beta$ is very small and there is a large variation in
$\phi(.)$. An alternative could be to add a constant bias to the
density function. We will not however consider this case here
and leave it for future work.
\end{remark}
\begin{remark}
Although we have derived the expression for $\lambda_i$'s
assuming that the agent density functions $f_i$ are
Gaussians, the above computations carry over to any definition
of $f_i$'s which are decreasing as a function of $\|p_i-q\|^2$.
\end{remark}

\noindent Setting the gradient $\frac{\partial d_2}{\partial p_i} = 0$ gives
\begin{equation}
 \begin{aligned}
  p_i &= C_{\mathcal{V}_i}^{\lambda} =
  \frac{\int_{\mathcal{V}_i} q \lambda_i(q) dq}
  {\int_{\mathcal{V}_i} \lambda_i(q) dq}
 \end{aligned}
\end{equation}
Thus the critical points of the $\mathcal{L}^2$ optimization
problem is described by
$p_i = \mathcal{C}_{\mathcal{V}_i}^{\lambda}$.
$\lambda_i$ defines a density on $\mathcal{V}_i$ and
$C_{\mathcal{V}_i}$ is the generalized centroid of
$\mathcal{V}_i$ corresponding to the density function $\lambda_i$.
We call the critical point defined by
$p_i = \mathcal{C}_{\mathcal{V}_i}^{\lambda}; \,\, i=1,2,\dots,N$
as a generalized centroidal voronoi configuration corresponding to
the $\lambda_i$'s.


\subsection{Control laws for single integrator agents}
\label{sec:1-3-1}
We assume that the agent dynamics are given by
\begin{equation}
 \dot{p}_i = u_i, \qquad i=1,2,\dots,N.
 \label{eqn:singleintdynamicsgencov}
\end{equation}
where $p_i \in \mathcal{Q}^q$ is the position of agent $i$ and
$u_i \subset \mathbb{R}^n$ is the control input. Then we have
the following result.
\begin{theorem}
 Consider $N$ agents deployed in a convex and bounded
 region $\mathcal{Q} \subset \mathbb{R}^n$.
 Let the agent dynamics be given by
 (\ref{eqn:singleintdynamicsgencov}). Then the control law
 given by
 \begin{equation}
  u_i = -K_p (p_i - \mathcal{C}_{\mathcal{V}_i}^{\lambda})
 \end{equation}
 with $0 < \gamma \leq \beta$ and $K_p > 0$,
 drives the agents to a minimum of the cost function
 (\ref{eqn:d2}) which is the generalized centroidal voronoi
 configuration with respect to the $\lambda_i$'s.
\end{theorem}

\begin{proof}
Consider the function
\begin{equation*}
 V(t) = d_2(\phi,f_P)
\end{equation*}
Taking the derivative,
\begin{equation*}
 \centering
 \begin{aligned}
  \dot{V} &= \sum_{i=1}^N
  \frac{\partial d_2}{\partial p_i}^{\trp}
  \dot{p}_i
  = \sum_{i=1}^N
  \frac{\partial d_2}{\partial p_i}^{\trp} u_i \\
  &= \sum_{i=1}^N
  \frac{4 \gamma}{\sigma^2} M_{\mathcal{V}_i}^{\lambda} (p_i -
  \mathcal{C}_{\mathcal{V}_i}^{\lambda})^{\trp} u_i
 \end{aligned}
\end{equation*}
Substituting the control law, we get
\begin{equation*}
 \begin{aligned}
  \dot{V} &= - \sum_{i=1}^N \frac{4 \gamma}{\sigma^2} K_p
  M_{\mathcal{V}_i}^{\lambda} \| p_i-
  \mathcal{C}_{\mathcal{V}_i}^{\lambda} \|^2
 \end{aligned}
\end{equation*}
$V$ is continuously differentiable on the compact set
$\mathcal{Q}$, and $\mathcal{Q}$ is positively invariant with
respect to the closed loop dynamics.
Since $\dot{V} \leq 0$, from LaSalle invariance principle
[\citet{khalil2002nonlinear}], we conclude
that the trajectories converge to the largest invariant set in
$\{ p_i :
\| p_i - \mathcal{C}_{\mathcal{V}_i}^{\lambda} \| = 0 \}$
which is the set itself. Since the cost function $V$ is
decreasing with time, we see that it converges to a minimum.
This concludes the proof.
\end{proof}

\begin{remark}
The above control law is similar to the control law derived
for the locational optimization case discussed previously
[\citet{cortes2004coverage}] except that in defining the
generalized centroid, $\phi(q)$ is replaced by the modified
function $\lambda_i(q) \,\, \mbox{for} \,\, q \in \mathcal{V}_i$.
\end{remark}

\subsection{Density function}
In the sequel, we assume that the density function
is unknown as will be the case in most practical scenarios.
We will also assume that it can be linearly parameterized
as [see \citet{Schwager2009}]
\begin{equation}
 \phi(q) = \mathcal{K}(q)^{\trp} a
\end{equation}
where
$\mathcal{K}(q)^{\trp} =
\left[ \mathcal{K}^1(q), \mathcal{K}^2(q), \dots,
\mathcal{K}^p(q) \right]$ is a vector of known basis functions
evaluated at $q \in \mathcal{Q}$ and $a \in \mathbb{R}^p$ is an
unknown constant parameter vector. Here
$\mathcal{K}^i : \mathcal{Q} \to \mathbb{R}_+$ and $a > 0$.
Each agent estimates the parameter $a$ to form the estimate for
$\phi(.)$. Each agent is also capable of measuring the value of
$\phi(.)$ at its current location. We also assume that the
parameters are lower bounded, i.e. $a^i > a_{\mbox{\small min}}$
where $a^i$ is the $i$-th component of $a$.

\section{Control laws for Differential Drive Robots}
\label{sec:control}
In this section, we introduce the model for the differential drive
robots, and derive the adaptive control laws.

\subsection{Robot Model}
\label{sec:control-robotmodel}
The robot model is given by (see figure (\ref{fig:robotmodel}))
\begin{figure}
 \centering
 \subfloat{
 \includegraphics
 [scale=0.25]
 {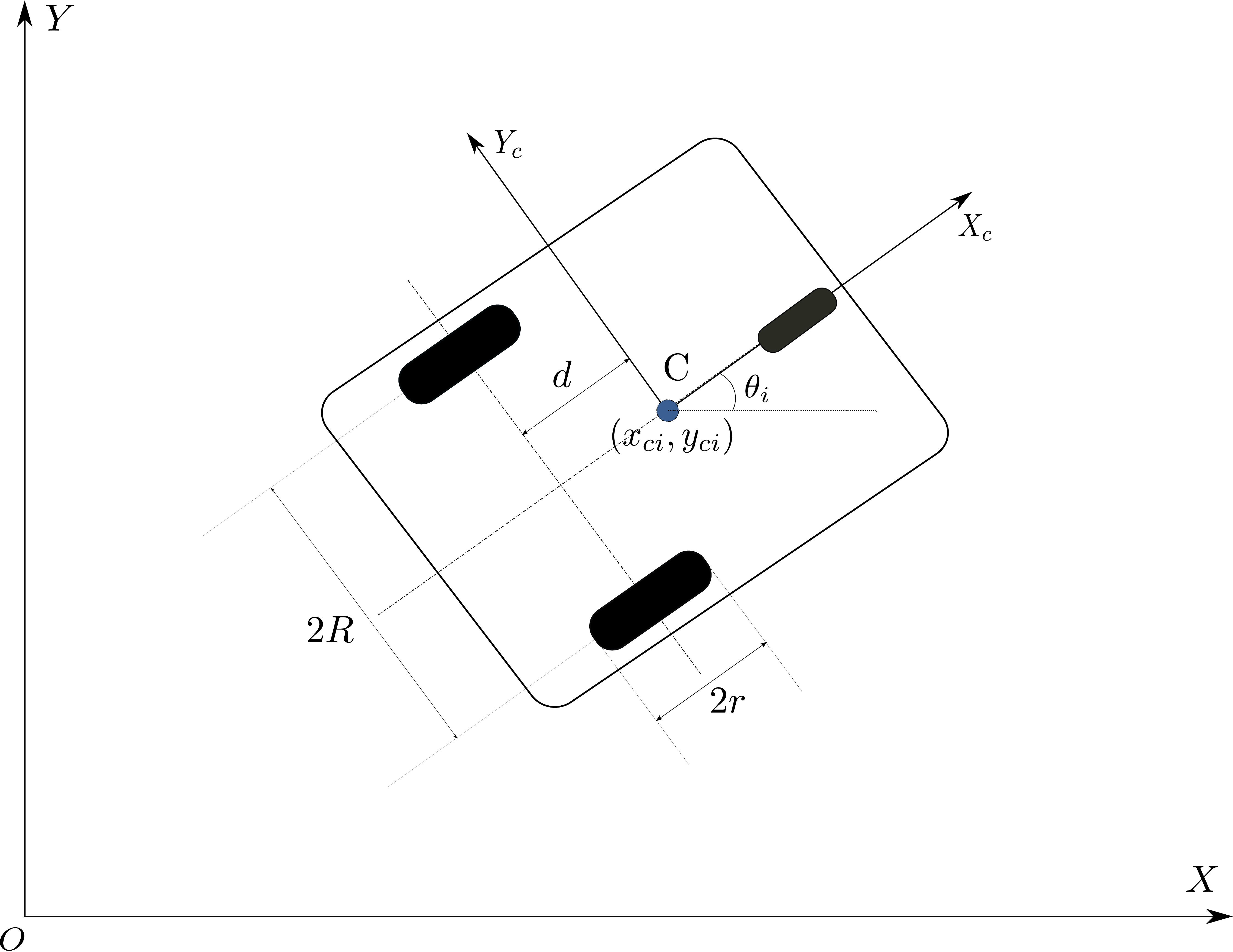}}
 \caption{Robot Model}
 \label{fig:robotmodel}
\end{figure}
\begin{equation}
 \left[ \begin{array}{c}
 \dot{x}_i \\
 \dot{y}_i \\
 \dot{\theta}_i
 \end{array} \right] =
 \left[ \begin{array}{cc}
 \mbox{cos}(\theta_i) & -d . \mbox{sin}(\theta_i) \\
 \mbox{sin}(\theta_i) & d . \mbox{cos}(\theta_i) \\
 0 & 1
 \end{array} \right]
 \left[ \begin{array}{c}
 u_i \\
 \omega_i
 \end{array} \right]
 \label{eqn:robotmodel}
\end{equation}
Here, $p_i = (x_i, y_i)^{\trp}$ is the position of the centre
point of robot $i$ (this corresponds to the point which is
tracked by the localization system). $\theta_i$ is the
orientation of robot $i$. The control input is given by
$v_i = (u_i,\omega_i)^{\trp}$ where $u_i$ and $\omega_i$
are the linear and angular velocity commands. $d$ is the distance
of the centre point from the wheel axis.

\subsection{Adaptive Control Law}
\label{sec:control-law}
The control law derivation is similar to that in
\citet{Schwager2009,rihab2018a}. We denote the estimate of
parameter $a$ of agent $i$ at time $t$ by $\hat{a}_i(t)$ and
the measurement of $\phi(.)$ by the agent $i$ by
$\phi_i(t) = \phi(p_i(t))$ where $p_i(t)$ is the position of
agent $i$ at time $t$. Now we define the following filters:
\begin{eqnarray}
 \dot{\Lambda}_i(t) &= -\alpha \Lambda_i(t) + \mathcal{K}_i(t)
 \mathcal{K}_i(t)^{\trp}\\
 \dot{\lambda}_i(t) &=
 -\alpha \lambda_i(t) + \mathcal{K}_i(t) \phi_i(t)
 \label{eqn:filtereqns}
\end{eqnarray}
where $\mathcal{K}_i(t) =
\mathcal{K}(p_i(t))$, $\Lambda_i(0) = {0}$, $\lambda_i(0) = {0}$.
\par \noindent The control and adaptation laws are then given by
\begin{align}
 \label{eqn:robotcontrol_l2}
 v_i &= -k_1 \left[S_i(q)\right]^{-1}
 (p_i-\hat{C}_{\mathcal{V}_i}^{\lambda}) \\
 \label{eqn:adaptation1_l2}
 \dot{\hat{a}}_i &= \Gamma(b_i - I_{\beta_i}b_i)
\end{align}
\begin{equation}
 \begin{aligned}
  b_i &= -k_2 \int_{\mathcal{V}_i} e^{-\frac{\|p_i-q\|^2}{\sigma^2}}
  \mathcal{K}(q) (q-p_i)^T dq \,\, S_i(q) v_i \\
  & \qquad - \gamma(\Lambda_i \hat{a}_i - \lambda_i)
 \end{aligned}
 \label{eqn:adaptation2_l2}
\end{equation}
where $k_1, k_2 > 0$ are positive gains, $S(q)$ is the matrix
\begin{equation}
 \centering
 S(q) =
 \left[ \begin{array}{cc}
 \mbox{cos}(\theta_i) & -d . \mbox{sin}(\theta_i) \\
 \mbox{sin}(\theta_i) & d . \mbox{cos}(\theta_i)
 \end{array}  \right]
\end{equation}
We can now state the following theorem:
\begin{theorem}
 \label{thm:l2differentialdrive}
 Consider $N$ differential drive robots deployed in the region
 $\mathcal{Q}$ for covering the region. Assume that the robots
 implement the control law \eqref{eqn:robotcontrol_l2} and the
 adaptation law \eqref{eqn:adaptation1_l2} and
 \eqref{eqn:adaptation2_l2}. Then the following statements hold:
 \begin{enumerate}
  \item $\lim_{t \to \infty}
  \|p_i - \hat{C}_{\mathcal{V}_i}^{\lambda}\| = 0$,
  \item $\lim_{t \to \infty} \|v_i\| = 0$,
  \item $\lim_{t \to \infty}
  \mathcal{K}_i(\tau) \tilde{a}_i(t) = 0$
  for $ \forall \tau \,\, \mbox{s.t} \,\,\, t-T<\tau<t$,
  where $T>0$.
 \end{enumerate}
 for all $i=1,2,\dots,N$.
\end{theorem}
\begin{proof}
Consider the lyapunov function
\[
 V = d_2(\phi,f_P) +
 \frac{1}{2} \sum_{i=1}^N \tilde{a}_i^{\trp} \Gamma^{-1} \tilde{a}_i
\]
Taking the derivative,
\begin{align*}
 \dot{V} &= \sum_{i=1}^N \frac{\partial d_2}{\partial p_i}^{\trp}
 \dot{p}_i +
 \sum_{i=1}^N \tilde{a}_i^{\trp} \Gamma^{-1} \dot{\hat{a}}_i
\end{align*}
Using \eqref{eqn:derivd2_gaussian}, we get
\begin{align*}
 \dot{V} &= \sum_{i=1}^N \frac{4 \gamma}{\sigma^2}
 M_{\mathcal{V}_i}^{\lambda}
 (p_i - C_{\mathcal{V}_i}^{\lambda})^{\trp} \dot{p}_i +
 \sum_{i=1}^N \tilde{a}_i^{\trp} \Gamma^{-1} \dot{\hat{a}}_i
\end{align*}
Using the definitions of $M_{\mathcal{V}_i}^{\lambda}$ and
$C_{\mathcal{V}_i}^{\lambda}$,
\begin{align}
 \dot{V} &= \frac{4 \gamma}{\sigma^2}
 \sum_{i=1}^N \int_{\mathcal{V}_i} (p_i - q)^{\trp}
 \lambda_i(q) dq \, \dot{p}_i +
 \sum_{i=1}^N \tilde{a}_i^{\trp} \Gamma^{-1} \dot{\hat{a}}_i
 \label{eqn:proof1}
\end{align}
Now $\lambda_i(q)$ can be written as
\begin{align*}
 \lambda_i(q) &= e^{-\frac{\|p_i-q\|^2}{\sigma^2}} \left[ \phi(q) -
 \gamma \cdot e^{-\frac{\|p_i-q\|^2}{\sigma^2}} \right] \\
 &= e^{-\frac{\|p_i-q\|^2}{\sigma^2}} \left[ \mathcal{K}(q)^{\trp}
 (\hat{a}_i - \tilde{a}_i) -
 \gamma \cdot e^{-\frac{\|p_i-q\|^2}{\sigma^2}} \right] \\
 &= e^{-\frac{\|p_i-q\|^2}{\sigma^2}} \left[ \hat{\phi}(q) -
 \gamma \cdot e^{-\frac{\|p_i-q\|^2}{\sigma^2}} \right] \\
 & \quad - \tilde{a}_i^{\trp} \mathcal{K}(q)
 e^{-\frac{\|p_i-q\|^2}{\sigma^2}} \\
 &= \hat{\lambda}_i(q) - \tilde{a}_i^{\trp} \mathcal{K}(q)
 e^{-\frac{\|p_i-q\|^2}{\sigma^2}}.
\end{align*}
where $\hat{\lambda}_i(q) :=
e^{-\frac{\|p_i-q\|^2}{\sigma^2}} \left[ \hat{\phi}(q) -
\gamma \cdot e^{-\frac{\|p_i-q\|^2}{\sigma^2}} \right] $. \\
\noindent Using this in equation \eqref{eqn:proof1}, we get
\begin{align*}
 \dot{V} &= \frac{4 \gamma}{\sigma^2} \sum_{i=1}^N
 \int_{\mathcal{V}_i} (p_i - q)^{\trp}
 \hat{\lambda}_i(q) dq \, \dot{p}_i +
 \sum_{i=1}^N \tilde{a}_i^{\trp} \Gamma^{-1} \dot{\hat{a}}_i \\
 & \quad + \frac{4 \gamma}{\sigma^2} \sum_{i=1}^N
 \tilde{a}_i^{\trp} \int_{\mathcal{V}_i} \mathcal{K}(q)
 e^{-\frac{\|p_i-q\|^2}{\sigma^2}} (p_i - q)^{\trp} dq
 \, \dot{p}_i
\end{align*}
Using \eqref{eqn:robotmodel}, the control law
\eqref{eqn:robotcontrol_l2}, adaptation law
\eqref{eqn:adaptation1_l2} and simplifying, we get
\begin{align*}
 \dot{V} &= -\frac{4 \gamma}{\sigma^2} k_1 \sum_{i=1}^N
 \|p_i - \hat{C}_{\mathcal{V}_i}^{\lambda}\|^2
 - \sum_{i=1}^N \tilde{a}_i^{\trp} I_{\beta_i} b_i \\
 & \quad - \tilde{a}_i^{\trp} \int_0^t
 e^{-\alpha (t-\tau)} \mathcal{K}_i(\tau)
 \mathcal{K}_i(\tau)^{\trp} d \tau \, \tilde{a}_i.
\end{align*}
It can be shown that all three terms of $\dot{V}$ above are
non-positive [see \citet{Schwager2009}]. Since $V$ is bounded
below by zero and its time derivative is non-positive, it
follows that $\lim_{t \to \infty} V(t)$ is finite. This implies
that $\dot{V}$ is integrable. Using Barbalat's lemma, we can
conclude that $\lim_{t \to \infty} \dot{V} = 0$. Statements 1
and 3 of the theorem follow immediately. Statement 2 follows
from statement 1 and equation \eqref{eqn:robotcontrol_l2}.
The proof is thus complete.
\end{proof}

\begin{remark}
From the statements of theorem \ref{thm:l2differentialdrive}, we
also observe that $\lim_{t \to \infty} \dot{\theta}_i = 0$ for
all $i=1,2,\dots,N$.
\end{remark}

\begin{remark}
Since we intend to compare the performance of the $\mathcal{L}^2$
framework with the locational optimization problem, we mention
below the control and adaptation laws for the locational
optimization case [see \citet{rihab2018a} for more details on
the same] which is derived similar to theorem
\ref{thm:l2differentialdrive}.
\begin{align}
 v_i &=
 -k_1 \left[S_i(q)\right]^{-1} (p_i-\hat{C}_{\mathcal{V}_i}) \\
 \dot{\hat{a}}_i &= \Gamma(b_i - I_{\beta_i}b_i)
 \label{eqn:controllaw_locopt}
\end{align}
\begin{equation}
 \begin{aligned}
  b_i &= -k_2 \int_{\mathcal{V}_i}
  \mathcal{K}(q) (q-p_i)^T dq \,\, S_i(q) v_i \\
  &\qquad - \gamma(\Lambda_i \hat{a}_i - \lambda_i)
 \end{aligned}
 \label{eqn:adaptationlaw_locopt}
\end{equation}
\end{remark}

\subsection{Consensus for improving Parameter Convergence}
\label{sec:control-consensus}
From the proof of theorem \ref{thm:l2differentialdrive}, it can
be observed that the parameter estimate of $\hat{a}_i$ converges
to the true value $a$ if the matrix
\[
 \lim_{t \to \infty} \int_0^t e^{-\alpha (t - \tau)}
 \mathcal{K}_i(\tau) \mathcal{K}_i(\tau)^{\trp} d\tau
\]
is positive definite.
In \citet{Schwager2009}, a consensus term was proposed to be
included in the adaptation law to improve parameter convergence.
It was shown that using a consensus term in the adaptation law
makes the parameter estimates of the agents converge to a common
value and thus also weakens the sufficient richness condition
required for parameter convergence. The modified adaptation law
is given by:
\begin{align}
 \dot{\hat{a}}_i &= \Gamma \left(b_i - I_{\beta}b_i \right) \\
 b_i &= -F_i(p_i) - \gamma (\Lambda_i \hat{a}_i - \lambda_i) -
 \zeta \sum_{j=1}^N l_{ij} (\hat{a}_i - \hat{a}_j)
\end{align}
where $F_i(p_i)$ is the integral term in the adaptation law.
The underlying graph used for consensus is the delaunay graph
where the agents which share an edge of voronoi partition have
the corresponding coefficients $l_{ij}$ non-zero.
In \cite{Schwager2009} the authors propose that $l_{ij}$ be equal
to the length of the shared voronoi edge
($|\mathcal{V}_i \cap \mathcal{V}_j|$) between agent $i$ and $j$.
An important consequence of using consensus based adaptation law
is [see \cite{Schwager2009}].
\begin{corollary}[Corollary 2, \cite{Schwager2009}]
Using the consensus adaptation law, in addition to the
convergence of position and velocity, if the agent paths are
such that
\[ \sum_{i=1}^N \int\limits_0^t e^{-\alpha (t-\tau)}
\mathcal{K}_i(\tau) \mathcal{K}_i(\tau)^{\trp} d \tau \]
is positive definite, each agent's parameter estimate converges
to the true value of the parameter.
\end{corollary}
The above condition is weaker since with consensus, the positive
definiteness condition is over sum of trajectories of all agents
as opposed to the individual trajectories for each of the
agents.

\subsubsection{Directed Consensus}
It can be seen that the agents whose trajectories follow a
certain path estimates certain parameters to large accuracy
whereas for other parameters they have poor estimates. This
can be observed from the adaptation term
$-(\Lambda_i \hat{a}_i - \lambda_i)$ which means that the error
between the measured and the estimated value of $\phi(.)$ is
weighted by the corresponding regressor element
$\mathcal{K}_i^{(j)}$ for updating the corresponding parameter
estimate $\hat{a}_i^{(j)}$. Thus if the agent trajectory is such
that the regressor element always takes a low value, then the
corresponding parameter estimate is also very poor. This means
that using a consensus term can sometimes reduce the accuracy
and/or convergence speed of parameter estimates of those agents
which are otherwise able to accurately estimate the parameter.
\par Based on the above observation, we propose a modified
consensus law. Corresponding to each parameter $a^{(j)}$, we
construct a directed sub-graph $\mathcal{G}^{(j)}(t)$ of the
delaunay graph $\mathcal{G}(t)$ as follows: a directed edge
between voronoi neighbours $i$ and $l$ exists if
\begin{equation}
	\mathcal{K}_i^{(j)}(t) \geq \mathcal{K}_l^{(j)}(t).
	\label{eqn:Ki_condition_consensus}
\end{equation}
The weights for the directed edges are taken as constant. This
protocol creates a seperate directed sub-graph of the
undirected delaunay graph corresponding to each parameter at
each time $t$. An illustration is shown in figure
\ref{fig:consensus_dirgraph}.

\begin{figure}
 \centering
 \includegraphics
 [scale=4.0]
 {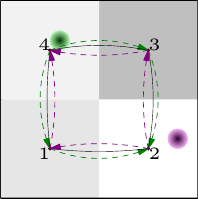}
 \caption{Illustration of directed graphs for consensus: $4$
 agents with $2$ parameters - the peak of the basis functions
 are shown in color with the directed graph for the corresponding
 parameter shown in the same color; the black color shows the
 original delaunay graph.}
 \label{fig:consensus_dirgraph}
\end{figure}

\begin{lemma}
If the delaunay graph $\mathcal{G}(t)$ is connected and the
basis functions in $\mathcal{K}(\cdot)$ are radial functions
(i.e., the functions have their peak value at some point and
the value reduces with distance from that point), then
the directed graphs $\mathcal{G}^{(j)}(t)$ for each $j$ has a
rooted tree. The root of the tree is the agent having the maximum
value of $\mathcal{K}_i^{(j)}(t)$ among all agents
$i=1,2,\dots,N$.
\end{lemma}

\begin{proof}
For any $j$ and each pair of agents $(i,l)$ which are voronoi
neighbours, we see from condition
\eqref{eqn:Ki_condition_consensus} that there is always a
directed edge either from $i$ to $l$ or from $l$ to $i$.
Then the agent with the largest value of
$\mathcal{K}_i^{(j)}(t)$, say $n_j$, will only have outgoing
edges, and the agent with the smallest value of
$\mathcal{K}_i^{(j)}(t)$ will only have incoming edges. Given
any node (agent) $l$, there always exists atleast one path in
the delaunay graph $\mathcal{G}(t)$ (since $\mathcal{G}(t)$
is connected) from agent $n_j$ to $l$. In particular,
there exists a path such that the sequence of nodes starting
from $n_j$ in the path are in decreasing order with respect
to the value of $\mathcal{K}^j(\cdot)$, i.e.,
$\mathcal{K}^j(p_{n_j}) \geq \dots \geq \mathcal{K}^j(p_{l})$.
This assertion can be proved as follows: Given node $n_j$, either
$n_j$ is connected to $l$, or it is not. If it is connected,
we are done. If it is not connected, there exists another node
which is at lower distance from node $n_j$ as compared to node
$l$. This new node is either connected to $l$, or there has
to exist another node which is at smaller distance from this
new node as compare to node $l$. This can be continued until
all such nodes are exhausted. If all such nodes are exhausted,
and the node $l$ is still not connected to one of these nodes,
it means that the graph is disconnected which is not possible.
Thus there exists some sequence of nodes in increasing order
of distance from the node $n_j$. This along with the fact that
the functions $\mathcal{K}^j(\cdot)$ are radial proves the
assertion.
Using \eqref{eqn:Ki_condition_consensus} it easily follows that
in the directed graph $\mathcal{G}^{(j)}(t)$, there always
exists a directed path from the root node $n_j$ to any other
node. This completes the proof.
\end{proof}

\noindent Now we modify the adaptation law with each parameter $j$
having seperate consensus law according to the directed sub-graph
$\mathcal{G}^{(j)}(t)$. Thus we have the following adaptation law
for agent $i$'s parameter estimate:
\begin{align}
 b_i &= -F_i(p_i) - \gamma (\Lambda_i \hat{a}_i - \lambda_i) -
 \zeta \sum_{j=1}^N L_{ij} (\hat{a}_i - \hat{a}_j) \\
 &= -F_i(p_i) - \gamma (\Lambda_i \hat{a}_i - \lambda_i) -
 \sum_{\alpha=1}^p \zeta \sum_{j=1}^N l_{ij}^{\alpha}
 \left(\hat{a}_i^{(\alpha)} - \hat{a}_j^{(\alpha)}\right)
 \label{eqn:adaptationlaw_modifiedconsensus}
\end{align}
where $L_{ij} = \mbox{diag}(\{l_{ij}^{\alpha}\}_{\alpha=1}^p)$
and $L^{\alpha} = \left[ l_{ij}^{\alpha} \right]$ is the
laplacian matrix for graph $\mathcal{G}^{(\alpha)}(t)$.

\begin{theorem}
Using the modified adaptation law
\ref{eqn:adaptationlaw_modifiedconsensus}, and assuming that all
the other conditions of thereom \ref{thm:l2differentialdrive}
hold, all the statements of the theorem
\ref{thm:l2differentialdrive} hold.
In addition
\begin{equation}
 \lim_{t \to \infty} (\hat{a}_i(t) - \hat{a}_j(t)) = 0,
\end{equation}
for all $i,j \in \{1,2,\dots,N\}$.
\end{theorem}

\begin{proof}
Proceeding the same way as in the proof of theorem
\ref{thm:l2differentialdrive}, we have an additional term in the
derivative of the lyapunov function $\dot{V}$, all other terms
remaining exactly the same. The term is given by
\[ T = -\sum_{i=1}^{N} \tilde{a}_i^{\trp} \sum_{j=1}^N L_{ij}
(\hat{a}_i - \hat{a}_j). \]
Simplifying this, we have
\begin{align*}
 T
 &= \sum_{i=1}^N \left[ \tilde{a}_i^{(1)} \tilde{a}_i^{(2)}
 \dots \tilde{a}_i^{(p)} \right]
 \sum_{j=1}^N L_{ij} \left[ \begin{array}{c}
 \hat{a}_i^{(1)} - \hat{a}_j^{(1)} \\
 \hat{a}_i^{(2)} - \hat{a}_j^{(2)} \\
 \vdots \\
 \hat{a}_i^{(p)} - \hat{a}_j^{(p)}
 \end{array} \right] \\
 &= \sum_{\alpha=1}^{p} \sum_{i=1}^N \sum_{j=1}^N
 \tilde{a}_i^{(\alpha)}
 l_{ij}^{\alpha} (\hat{a}_i^{(\alpha)} - \hat{a}_j^{(\alpha)}) \\
 &= \sum_{\alpha=1}^p \left[
 \tilde{a}_1^{(\alpha)} 
 \dots \tilde{a}_N^{(\alpha)} \right]
 \left[ \begin{array}{cccc}
 \sum_{i=2}^N l_{1i}^{\alpha}  &
 \dots  &  \dots  &  -l_{1N}^{\alpha} \\
 -l_{21}^{\alpha}  &
 \dots
 &  \dots  &  -l_{2N}^{\alpha} \\
 \vdots & \vdots & \vdots & \vdots \\
 -l_{N1}^{\alpha} & \dots & \dots &
 \vdots
 \end{array} \right]
 \left[ \begin{array}{c}
 \hat{a}_1^{(\alpha)} \\
 \hat{a}_2^{(\alpha)} \\
 \vdots \\
 \hat{a}_N^{(\alpha)}
 \end{array} \right] \\
 &= \sum_{\alpha=1}^p \tilde{a}^{{\alpha}^{\trp}}
 L^{\alpha} \hat{a}^{\alpha}
 = \sum_{\alpha=1}^p \hat{a}^{{\alpha}^{\trp}}
 L^{\alpha} \hat{a}^{\alpha}
\end{align*}
where $\hat{a}^{\alpha} = \left[ \hat{a}_1^{(\alpha)} \,
\hat{a}_2^{(\alpha)} \, \dots \hat{a}_N^{(\alpha)} \right]^{\trp}$.
\par \noindent Thus the term contributed by the consensus term in
$\dot{V}$ is non-positive. The term is also uniformly continuous.
The other terms in $\dot{V}$ remain non-positive and uniformly
continuous as in the proof of theorem
$\ref{thm:l2differentialdrive}$. Thus using Barbalat's lemma,
all the statements of the theorem $\ref{thm:l2differentialdrive}$
holds. In addition, we have that $\lim_{t \to \infty}
\hat{a}^{{\alpha}^{\trp}} L^{\alpha} \hat{a}^{\alpha} = 0$ for
each $\alpha$. Since $L^{\alpha}$ is the laplacian matrix of the
directed graph $\mathcal{G}^{(\alpha)}$, we have that
$\lim_{t \to \infty} \hat{a}^{\alpha} = c_{\alpha}\mathbf{1}$
for each $\alpha$ ($c_{\alpha}$ is some positive constant),
i.e. the agents achieve consensus on the parameter values and
the theorem holds.
\end{proof}

\section{Hardware Implementation and Experimental Results}
\label{sec:hardware}
In this section, we discuss the details of hardware
implementation as well as the experimental results.

\subsection{Experiment Setup}
\label{sec:hardware-setup}
The experimental setup consists of five differential drive robots,
based on the turtlebot3 platform, with OpenCR 1.0 controller
module and Raspberry Pi 3 module mounted on each of them.

\subsubsection{Workspace and Localization System}
The workspace where the robots move is a flat
$4 \mbox{m} \times 4 \mbox{m}$ square region.
For localization of robots, we use the motion-capture system
from Optitrack. The system comprises of 16 cameras with infrared
sensors which detect the markers fixed atop the robots.
A proprietary software (Motive, by Optitrack) uses data
in the form of images captured by the cameras, performs
localization computations, and provides position data for all
the robots in the workspace. This data is streamed over the
local network using the Virtual Reality Peripheral Network
\footnote{https://github.com/vrpn/vrpn/wiki}
(VRPN) protocol. The localization system provides
millimeter-level precision at high frequencies upto 200 Hz.
See figure \ref{fig:exptsetup1} for the overall setup.
\begin{figure*}
 \centering
 \includegraphics[scale=0.75]{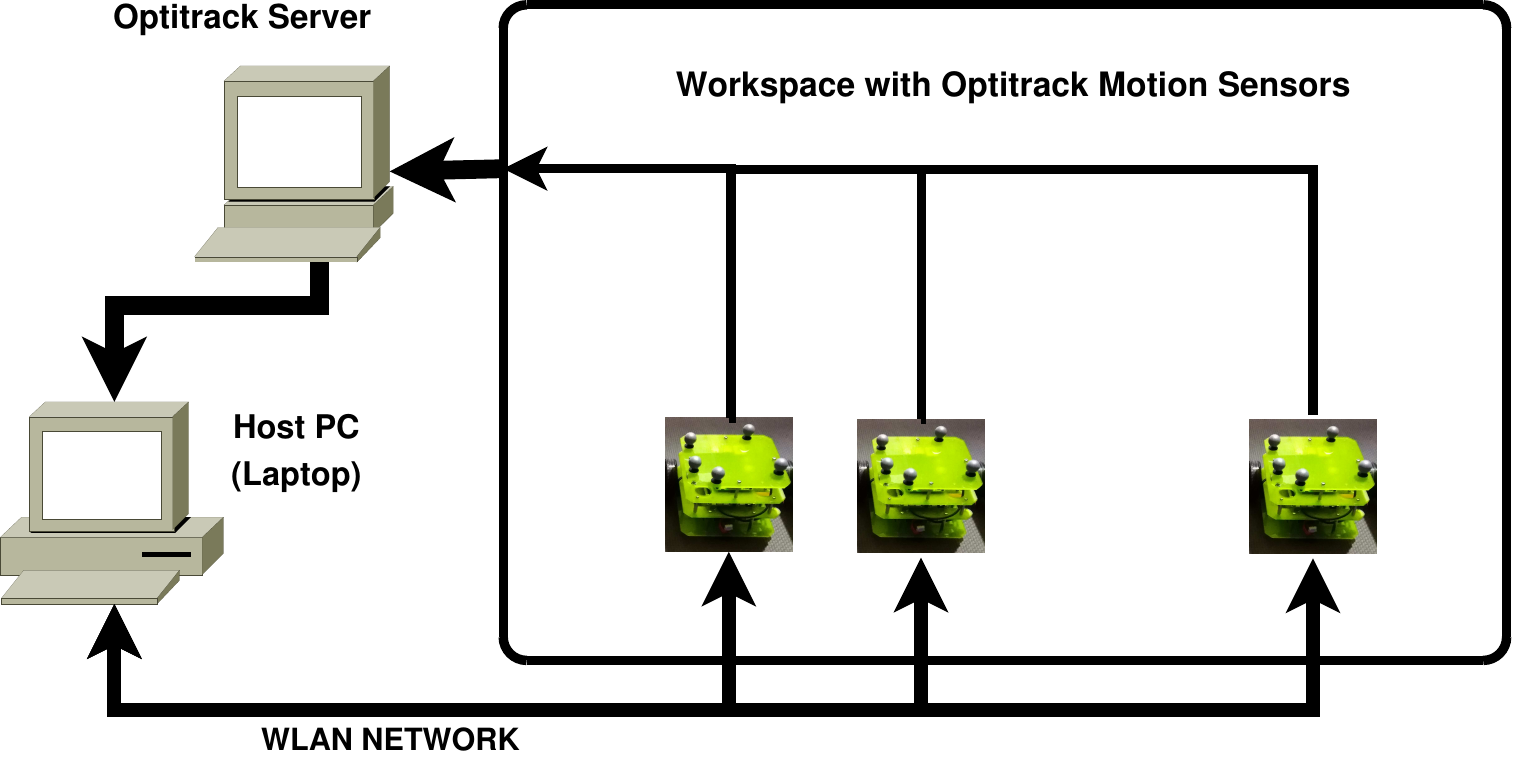}
 \caption{Experiment Setup}
 \label{fig:exptsetup1}
\end{figure*}
The robots communicate with the host PC via WLAN. ROS (Robot
Operating System) is used for the software implementation.
Each robot runs multiple ROS nodes which implements the coverage
algorithm, receives localization data from the Optitrack system
as well as communicate with other robots. The host PC runs the
ROS master node and subscribes to the position data from the
localization system (using VRPN protocol) which are then
distributed to the individual robots. Figure \ref{fig:ros_setup}
shows the overall software implementation using ROS.
\begin{figure*}
 \centering
 \includegraphics[scale=0.7]{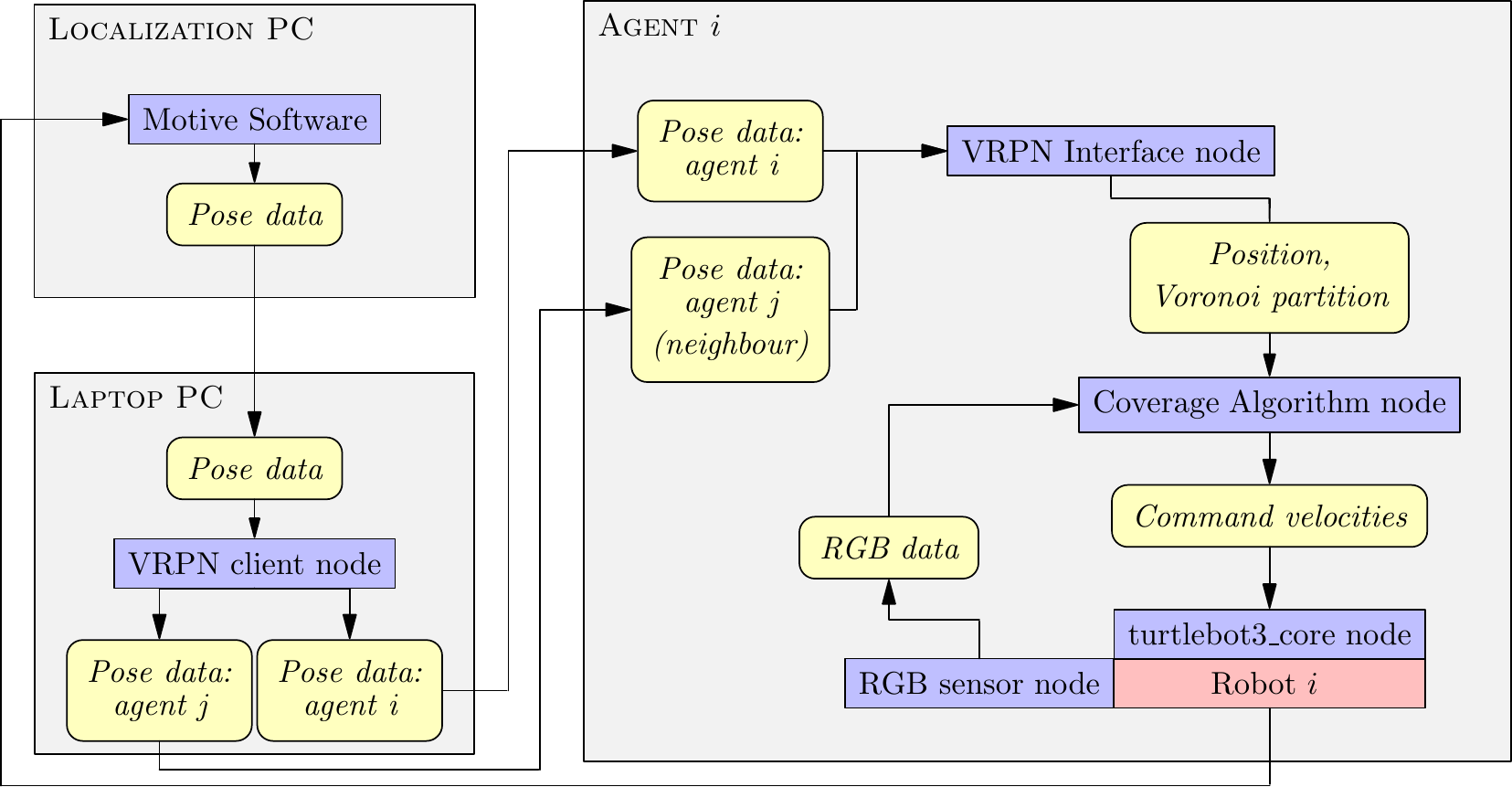}
 \caption{ROS Implementation: the blue blocks represent ROS
 nodes and the yellow blocks represent the data exchanged
 between nodes.}
 \label{fig:ros_setup}
\end{figure*}

\subsubsection{Sensors}
The density function is implemented as a light distribution
using Xiaomi smart bulbs. The Adafruit TCS34725 RGB sensors
are used for measuring the light intensity.

\subsection{Experimental Results}
\label{sec:hardware-results}
We do two sets of experiments: (1) The density function is
simulated, and (2) the density function is implemented using
white light sources, and the agents measure the light intensity
using light sensors. \\
The simulated density function case allows us to study the
performance of the coverage algorithm and parameter convergence
in better detail since there is no sensor noise and associated
issues. Implementing actual sources and sensors allows us to
evaluate how well the algorithms behave in the real world with
noisy sensors.
The values of various constants used in the simulation are given
in table \ref{tab:experiment_constants}.
\begin{table*}
 \centering
 \begin{tabular}{lcp{5.5cm}}
 \toprule 
 Parameter & Value & Description\\
 \cmidrule(r){1-1} \cmidrule(lr){2-2} \cmidrule(l){3-3} 
 \multicolumn{3}{l}{\textbf{\textit{Domain related}}} \\
 %
 borderx& [-2.0, 2.0, 2.0, -2.0] &
 x-coordinates of vertices of the domain \\
 bordery& [-2.0, -2.0, 2.0, 2.0] &
 y-coordinates of vertices of the domain \\
 %
 \midrule
 \multicolumn{3}{l}{\textbf{\textit{Density function related}}} \\
 centrex& [1.0, 1.0] & x-coordinates of centres of density fcn.\\
 centrey& [0.98, -0.8] & y-coordinates of centres of density fcn.\\
 $\sigma$& [0.6, 0.3] & std. deviation\\
 $a$& [85, 30] & true strengths to be estimated\\
 %
 \midrule
 \multicolumn{3}{l}
 {\textbf{\textit{Control and Adaptation gains}}} \\
 $k_1$& $0.1$ & controller gain\\
 $\Gamma$& $0.1 I$ & adaptation gain matrix\\
 %
 \midrule
 \multicolumn{3}{l}{\textbf{\textit{Loop rates}}} \\
 Control loop& $10$ Hz. & rate at which control loop runs\\
 Position update loop& $20$ Hz. &
 rate at which position data is available\\
 %
 \midrule
 \multicolumn{3}{l}{\textbf{\textit{Robot related}}} \\
 $d$& $0.05$ m & distance of the centre from wheel axis\\
 %
 \midrule
 \multicolumn{3}{l}{\textbf{\textit{Adaptation law related}}} \\
 paramInitValue& 10 & initial value of parameter estimates\\
 $\alpha$& 1.0 & filter parameter\\
 $\gamma$& 2 & measurement update gain\\
 $\zeta$& 1 & consensus related gain\\
 %
 %
 \bottomrule
 \end{tabular}
 \caption{Experiment related parameters}
 \label{tab:experiment_constants}
\end{table*}

\subsubsection{Simulated Density Function}
The true density function consists of two gaussian components.
The various constants related to the density function are given
in table \ref{tab:experiment_constants}. The trajectories,
average position error and the average velocity of the agents are
shown in figure \ref{fig:l2_adaptive_noconsensus_sim}. The
average position and velocity errors are given by
\begin{align*}
	e_p &= \sum_{i=1}^N \|p_i - \hat{C}_{\mathcal{V}_i}\|, \qquad
	e_v = \sum_{i=1}^N \|v_i\|.
\end{align*}
The corresponding plots for the locational optimization based
coverage is also shown in the figure for comparison. The initial
position error and velocity are higher for the locational
optimization case. The agent
parameter estimation errors are compared in figures
\ref{fig:par1error_noconsensus_sim} and
\ref{fig:par2error_noconsensus_sim} for the two parameters.
It can be seen that for three of the agents the
estimation errors for parameter $1$ drops significantly from
the initial value, with one agent able to estimate parameter $1$
accurately. The other two agents are not able to adapt for
parameter $1$. Similarly two of the agents are able to estimate
parameter $2$ better while for the rest of the agents the
parameter $2$ estimation error barely change from the initial
estimation error. It also appears that the $\mathcal{L}^2$
algorithm performs slightly better from figure
\ref{fig:par1error_noconsensus_sim}. This could be due
to the fact that the integral term in the $\mathcal{L}^2$
adaptation law \eqref{eqn:adaptation2_l2} is much smaller
(due to the presence of the exponential term)
than for the locational optimization adaptation law
\eqref{eqn:adaptationlaw_locopt}. This term can be viewed as
a coupling term between the coverage task (through the cost
function $d_2(\cdot)$) and the estimation task. The term
being smaller means that the estimate $\hat{a}_i$ is better
able to adapt through the measured error in $\phi(\cdot)$
given by the second term in the adaptation law
\eqref{eqn:adaptation2_l2}.
The parameter errors for adaptation with the consensus terms
are shown in figures \ref{fig:avgparerror_consensus_sim} and
\ref{fig:avgparerror_dirconsensus_sim}. We show the average
parameter estimation errors across all the agents for ease of
comparison since the agent estimation errors closely follow
each other due to the consensus term. From the plots we see
that overall, the parameter errors starts dropping faster in the
locational optimization case although towards the end the drop
in error becomes slower compared to the $\mathcal{L}^2$ case.
This could be due to the fact that
the initial velocity is higher for the locational optimization
case and thus it is able to initially move faster to regions
where significant measurements are available for adaptation.
Overall the final values of parameter estimates are
slightly better for the $\mathcal{L}^2$ case.
It can also be seen that the directional consensus algorithm
leads to a significantly faster convergence as opposed to the
undirected consensus algorithm.
\begin{figure}
 \centering
 \subfloat[Trajectories]{
 \includegraphics
 [scale=0.7]
 {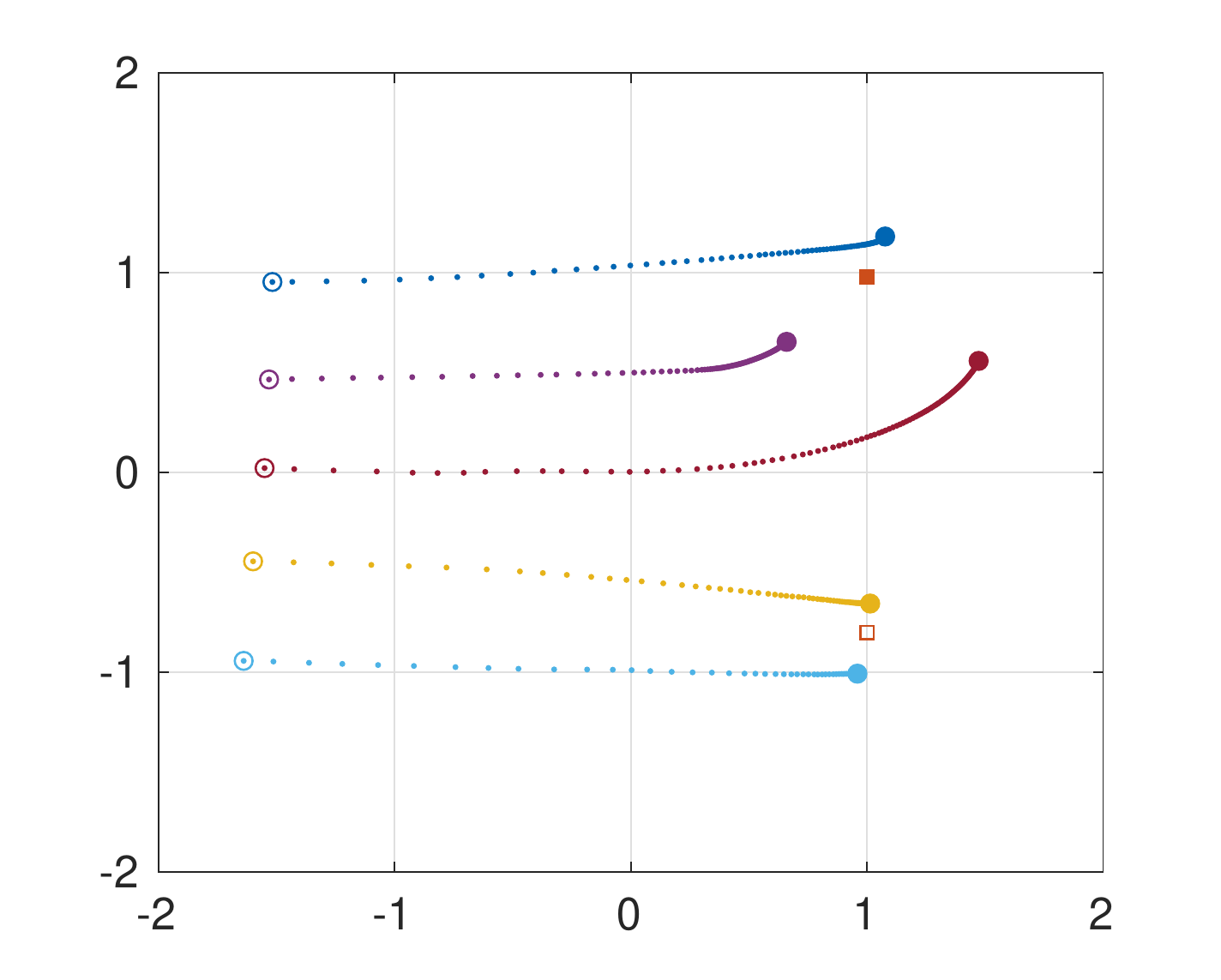}} \\
 \subfloat[Avg. Position error vs Time (sec).]{
 \includegraphics
 [height=0.2\textheight, width=0.9\textwidth]
 {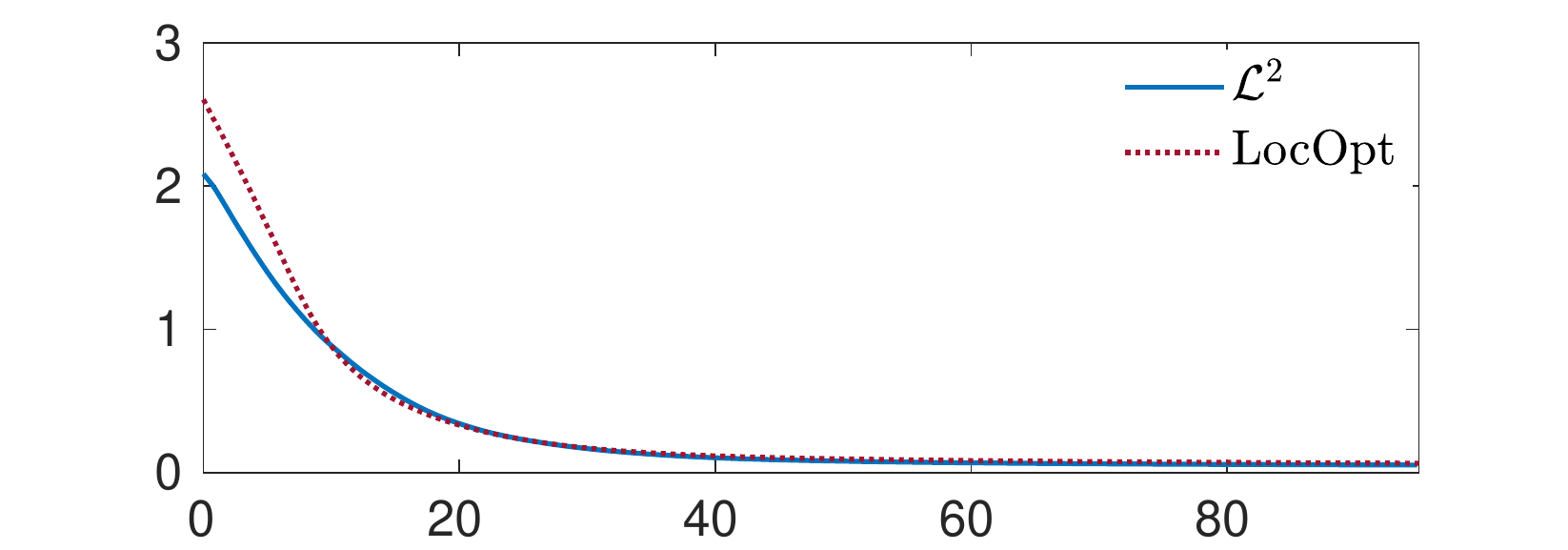}} \\
 \subfloat[Avg. Velocity vs Time (sec).]{
 \includegraphics
 [height=0.2\textheight, width=0.9\textwidth]
 {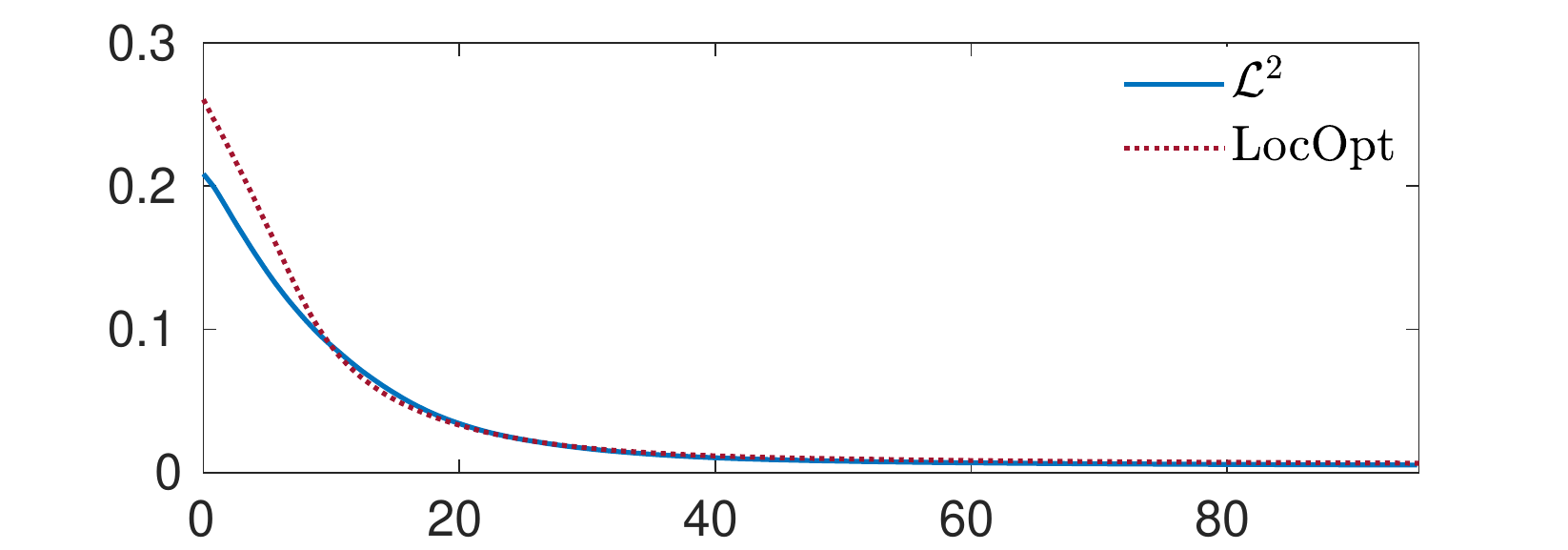}}
 \caption{Results for $\mathcal{L}^2$ coverage: simulated
 $\phi(\cdot)$.}
 \label{fig:l2_adaptive_noconsensus_sim}
\end{figure}

\begin{figure}
 \centering
 \subfloat[$\mathcal{L}^2$ coverage]{
 \includegraphics
 [height=0.22\textheight, width=0.5\textwidth]
 {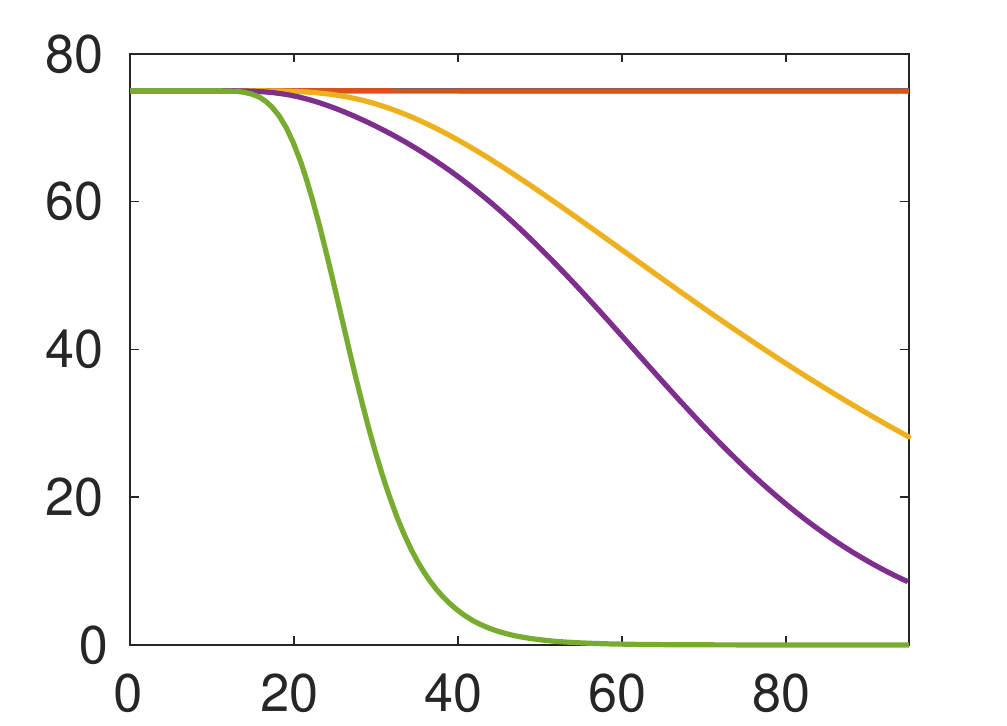}}
 \subfloat[LocOpt coverage]{
 \includegraphics
 [height=0.22\textheight, width=0.5\textwidth]
 {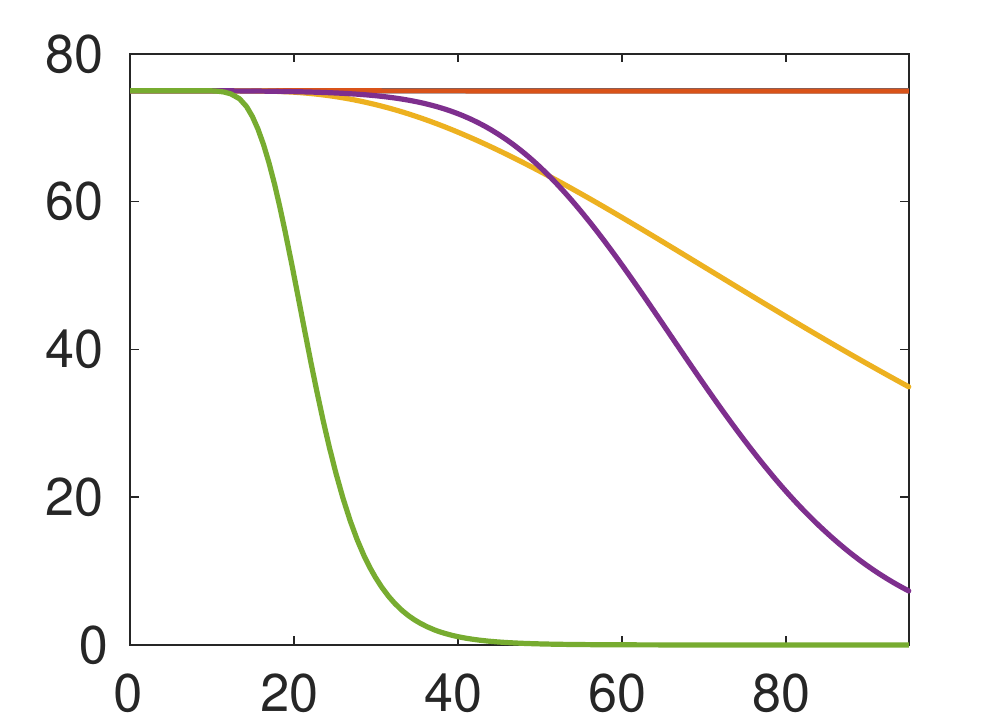}}
 \caption{Parameter $1$ estimation errors with time - No consensus.}
 \label{fig:par1error_noconsensus_sim}
\end{figure}

\begin{figure}
 \centering
 \subfloat[$\mathcal{L}^2$ coverage]{
 \includegraphics
 [height=0.22\textheight, width=0.5\textwidth]
 {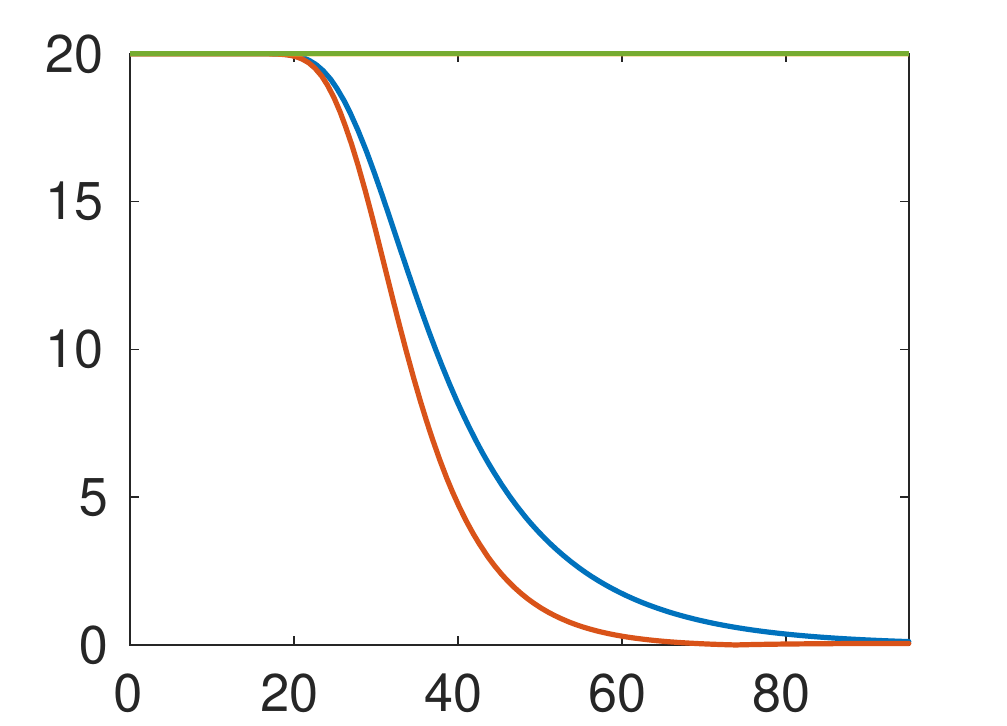}}
 \subfloat[LocOpt coverage]{
 \includegraphics
 [height=0.22\textheight, width=0.5\textwidth]
 {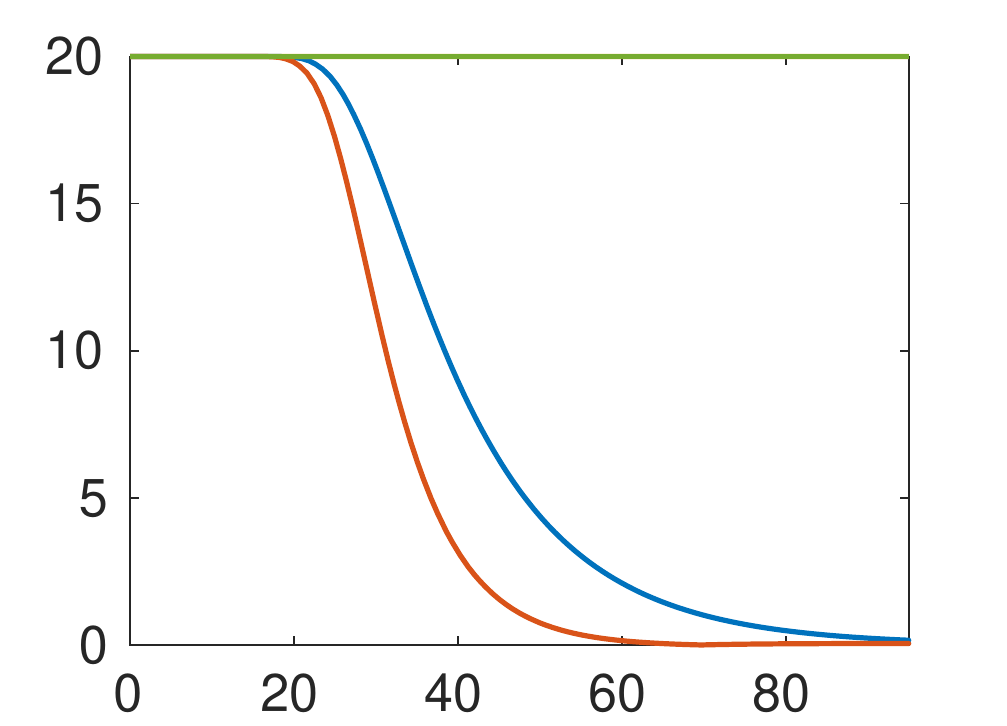}}
 \caption{Parameter $2$ estimation errors with time - No consensus.}
 \label{fig:par2error_noconsensus_sim}
\end{figure}

\begin{figure}
 \centering
 \subfloat[Parameter $1$]{
 \includegraphics
 [height=0.22\textheight, width=0.5\textwidth]
 {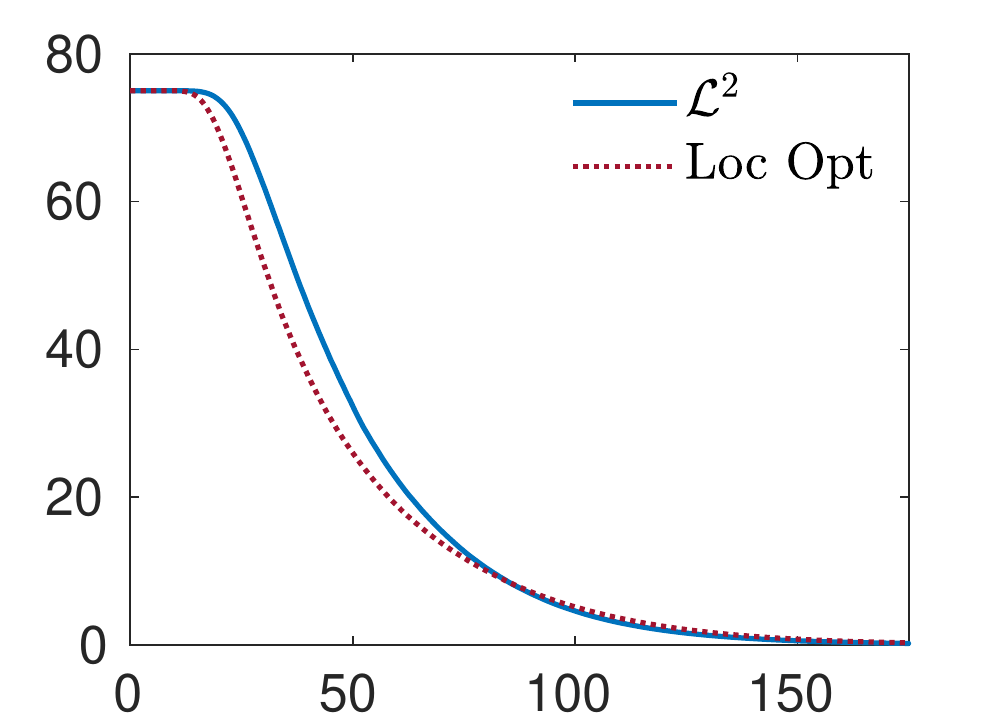}}
 \subfloat[Parameter $2$]{
 \includegraphics
 [height=0.22\textheight, width=0.5\textwidth]
 {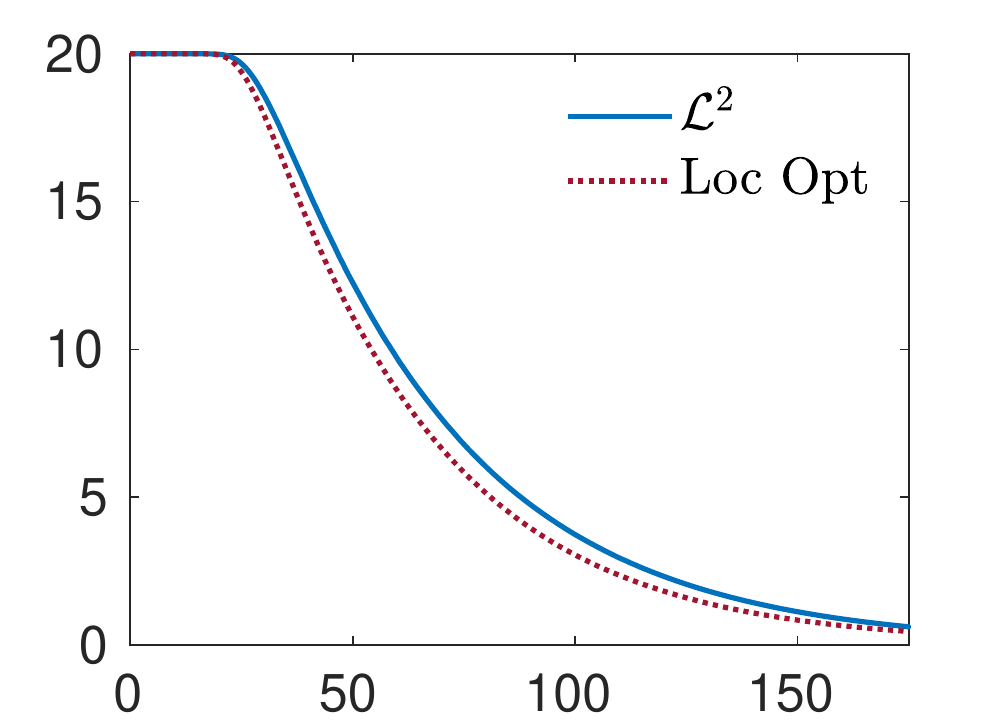}}
 \caption{Avg. parameter estimation errors with time - with
 consensus.}
 \label{fig:avgparerror_consensus_sim}
\end{figure}

\begin{figure}
 \centering
 \subfloat[Parameter $1$]{
 \includegraphics
 [height=0.22\textheight, width=0.5\textwidth]
 {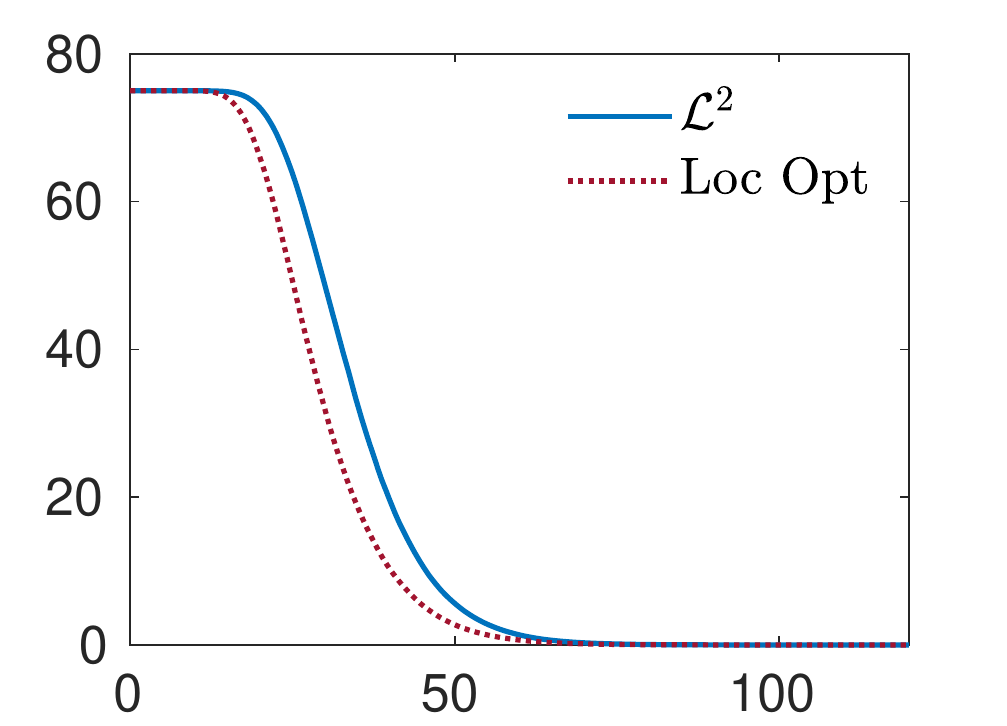}}
 \subfloat[Parameter $2$]{
 \includegraphics
 [height=0.22\textheight, width=0.5\textwidth]
 {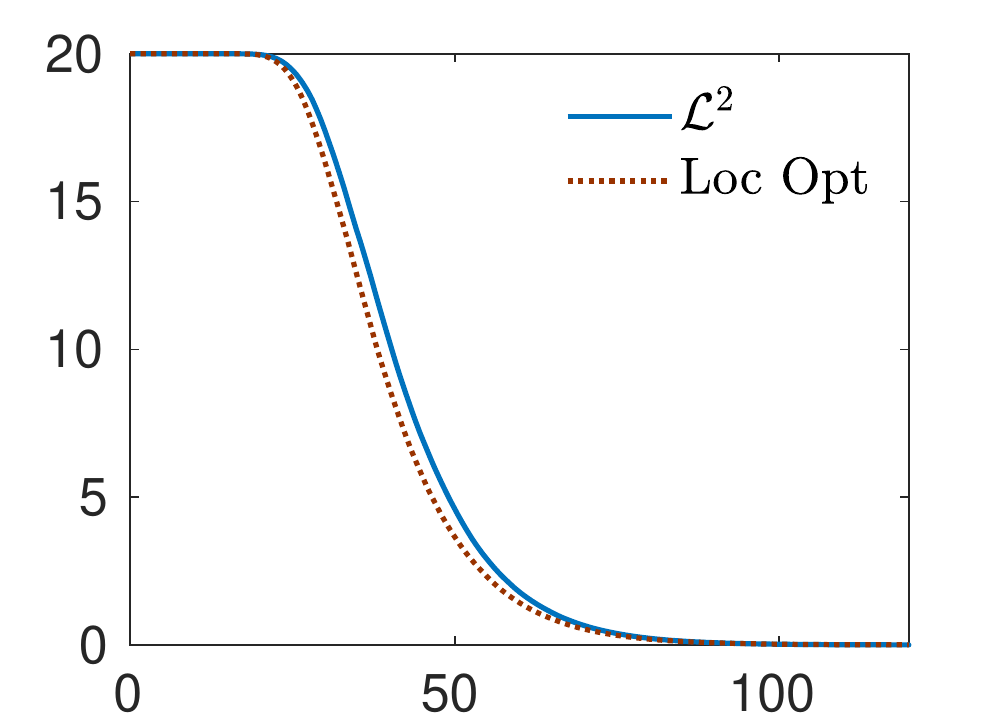}}
 \caption{Avg. parameter estimation errors with time - with
 directed consensus.}
 \label{fig:avgparerror_dirconsensus_sim}
\end{figure}

\subsubsection{Density Function implemented using Light Sources}
There are two light sources, one of high intensity and the other of
lower intensity. Each agent is equipped with TCS34725 RGB sensors
to measure the light intensity.
The trajectories, average position error and the average velocity
of the agents for $\mathcal{L}^2$ coverage are shown in figure
\ref{fig:l2_adaptive_noconsensus_rgb}. The plots also show
the results of locational optimization based coverage for comparison.
As with the results for simulated density function, we see that the
initial position error and velocity are larger for the locational
optimization case. The agent parameter estimation errors are
compared in figures \ref{fig:par1error_noconsensus_rgb} and
\ref{fig:par2error_noconsensus_rgb} for the two parameters with
no consensus term in the adaptation law.
The parameter errors for adaptation with the consensus terms
are shown in figures \ref{fig:avgparerror_consensus_rgb} and
\ref{fig:avgparerror_dirconsensus_rgb}. Overall the parameter
estimates using the $\mathcal{L}^2$ coverage framework seems to
be more accurate. It can also be seen that the direted consensus
leads to faster convergence of the parameter errors as expected.
\begin{figure}
 \centering
 \subfloat[Trajectories]{
 \includegraphics
 [scale=0.8]
 {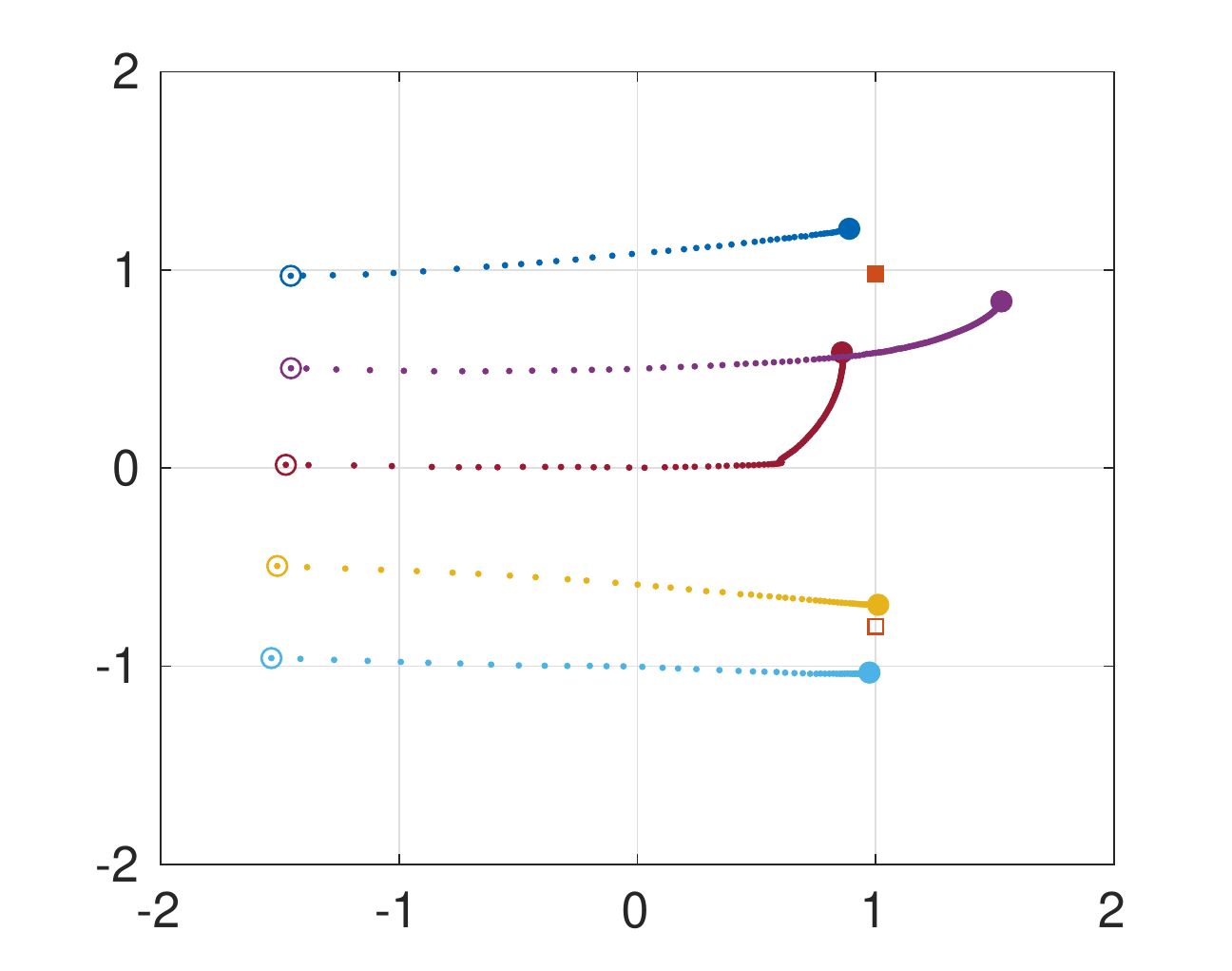}} \\
 \subfloat[Avg. Position error vs Time (sec).]{
 \includegraphics
 [height=0.2\textheight, width=0.9\textwidth]
 {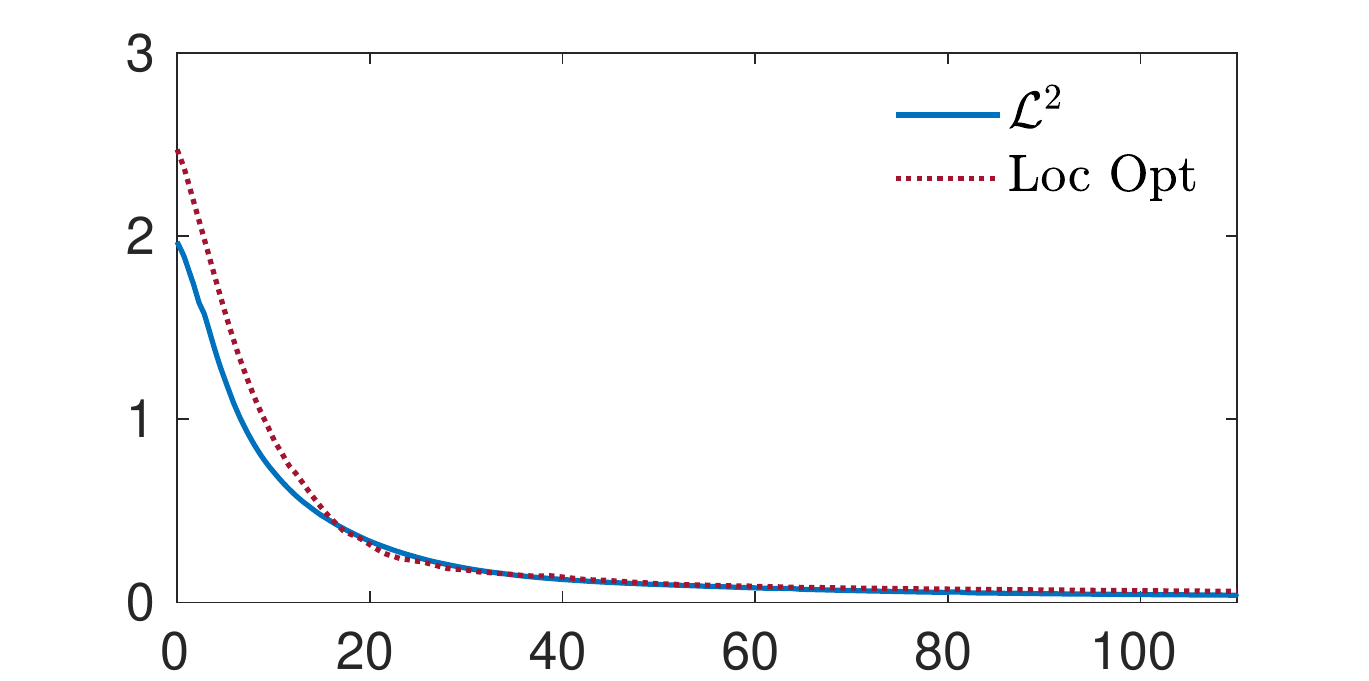}} \\
 \subfloat[Avg. Velocity vs Time (sec).]{
 \includegraphics
 [height=0.2\textheight, width=0.9\textwidth]
 {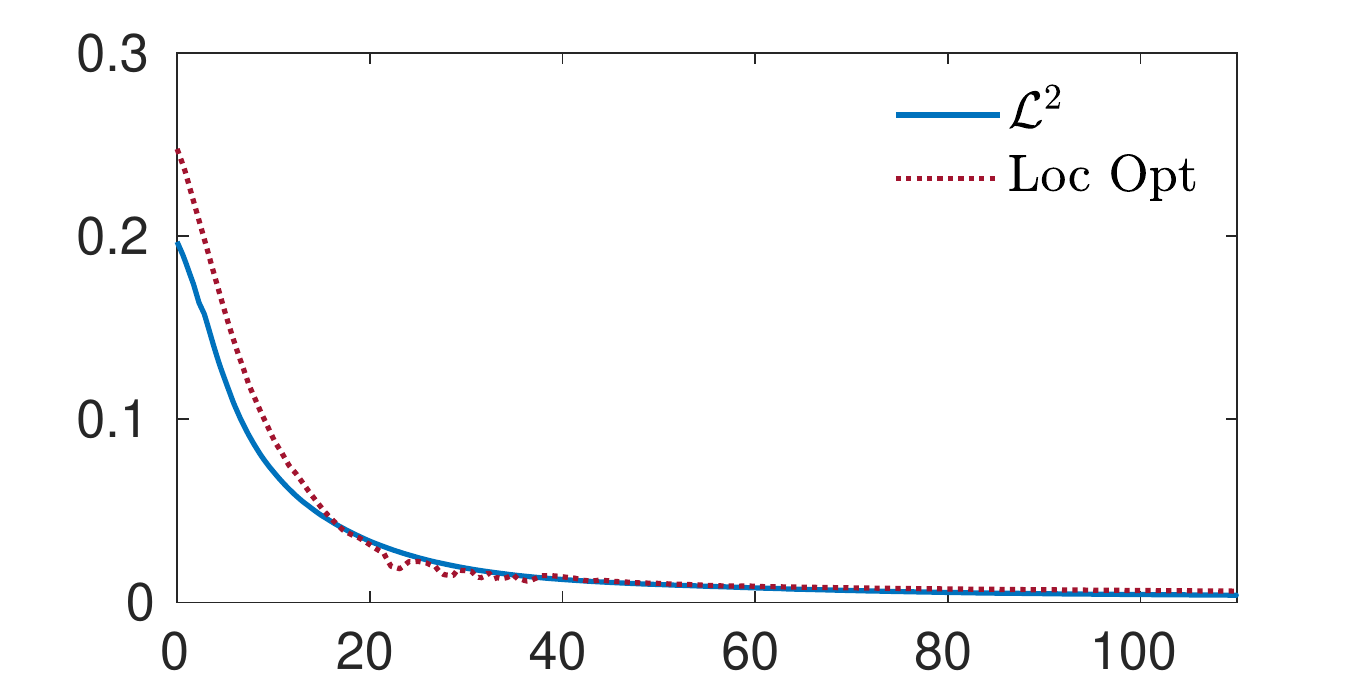}}
 \caption{Results for $\mathcal{L}^2$ coverage: RGB sensors.}
 \label{fig:l2_adaptive_noconsensus_rgb}
\end{figure}

\begin{figure}
 \centering
 \subfloat[$\mathcal{L}^2$ coverage]{
 \includegraphics
 [height=0.22\textheight, width=0.5\textwidth]
 {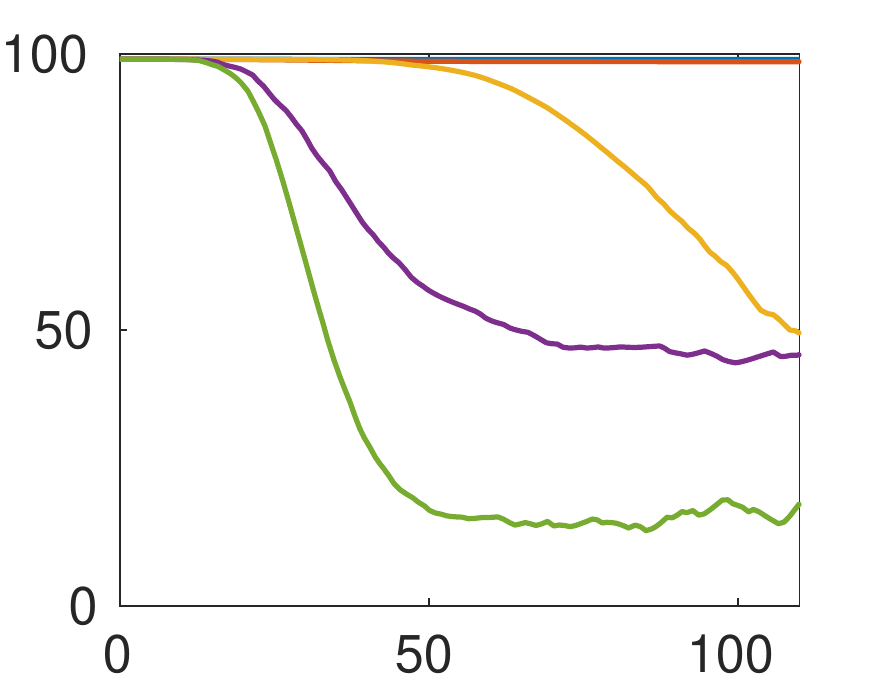}}
 \subfloat[LocOpt coverage]{
 \includegraphics
 [height=0.22\textheight, width=0.5\textwidth]
 {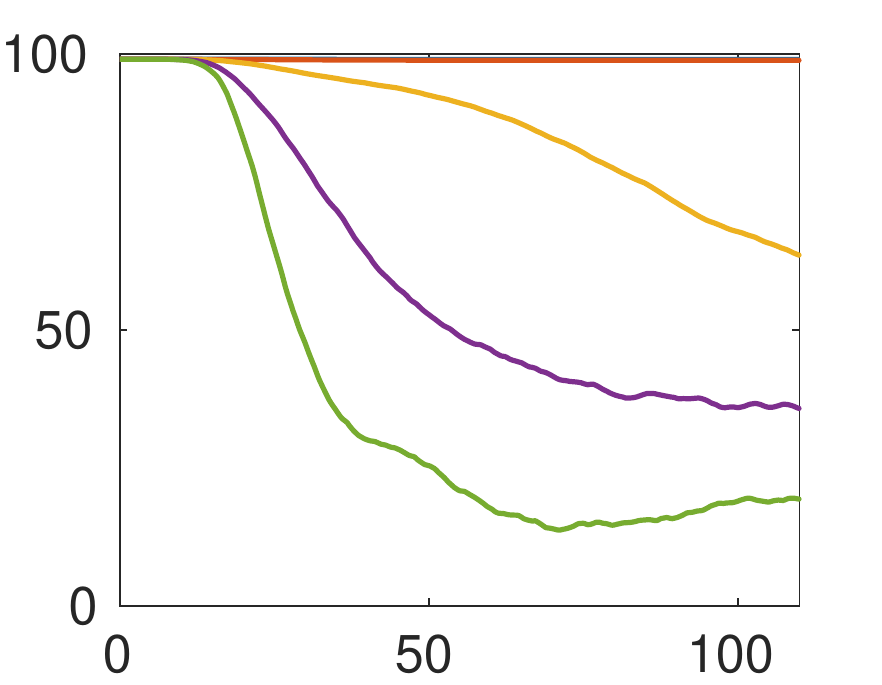}}
 \caption{Parameter $1$ estimation errors with time - No consensus.}
 \label{fig:par1error_noconsensus_rgb}
\end{figure}

\begin{figure}
 \centering
 \subfloat[$\mathcal{L}^2$ coverage]{
 \includegraphics
 [height=0.22\textheight, width=0.5\textwidth]
 {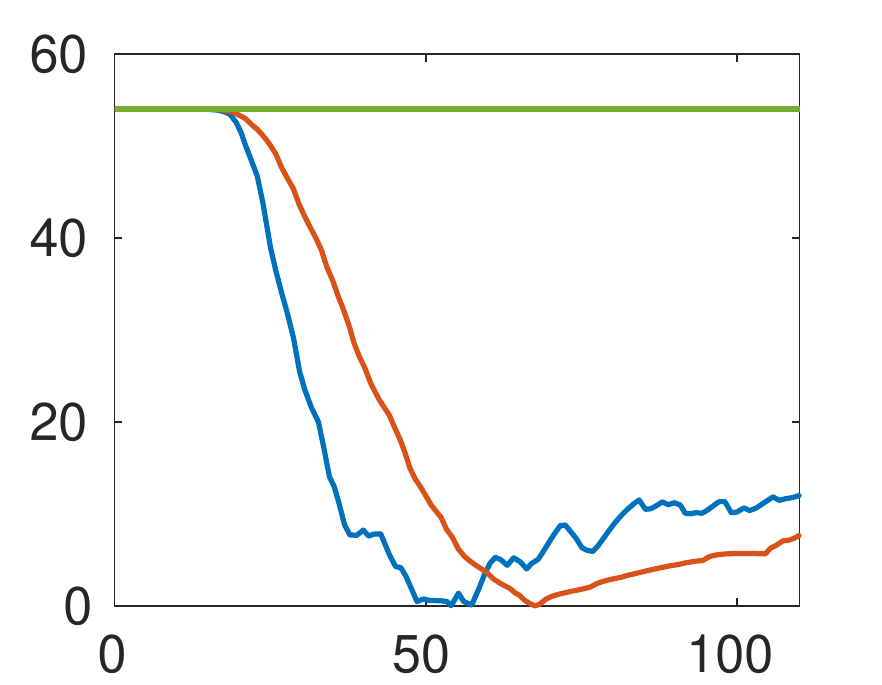}}
 \subfloat[LocOpt coverage]{
 \includegraphics
 [height=0.22\textheight, width=0.5\textwidth]
 {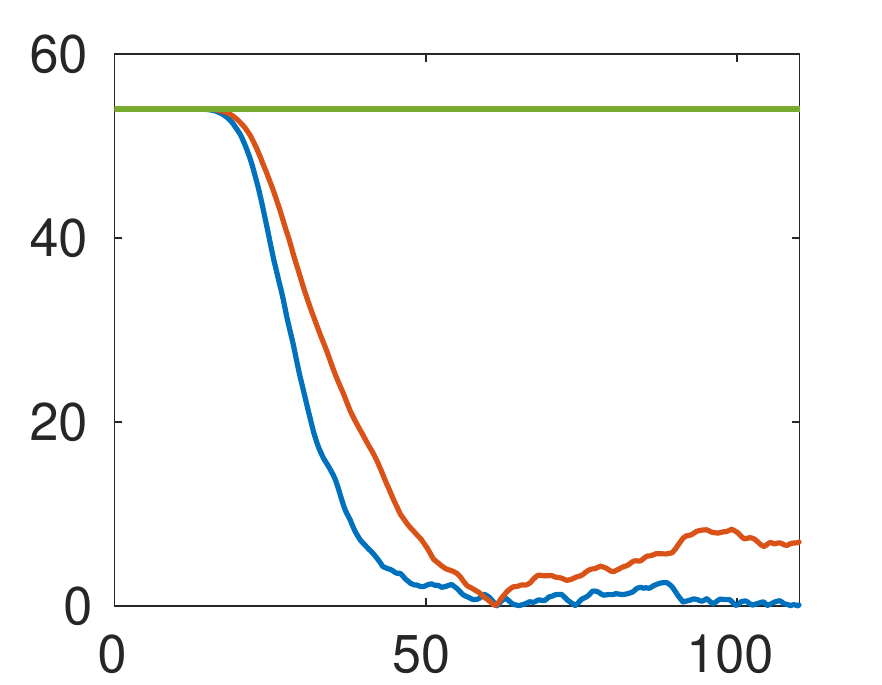}}
 \caption{Parameter $2$ estimation errors with time - No consensus.}
 \label{fig:par2error_noconsensus_rgb}
\end{figure}

\begin{figure}
 \centering
 \subfloat[Parameter $1$]{
 \includegraphics
 [height=0.22\textheight, width=0.5\textwidth]
 {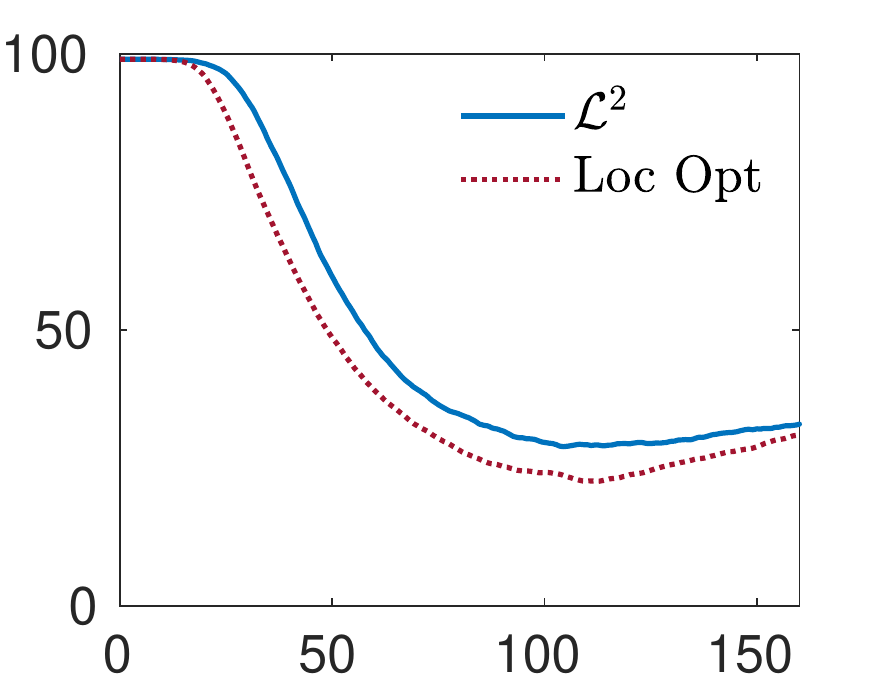}}
 \subfloat[Parameter $2$]{
 \includegraphics
 [height=0.22\textheight, width=0.5\textwidth]
 {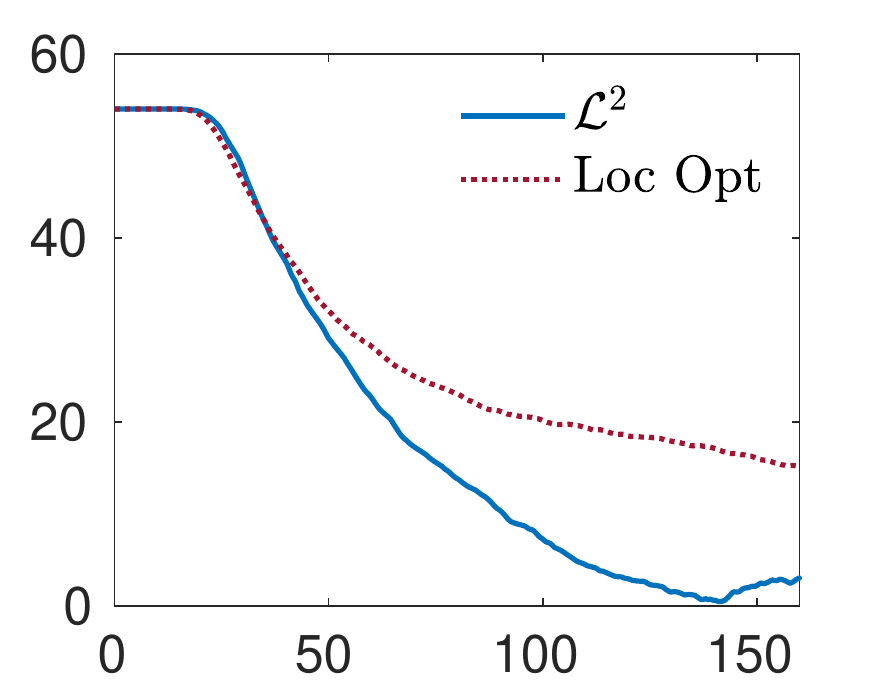}}
 \caption{Avg. parameter estimation errors with time - with
 consensus.}
 \label{fig:avgparerror_consensus_rgb}
\end{figure}

\begin{figure}
 \centering
 \subfloat[Parameter $1$]{
 \includegraphics
 [height=0.22\textheight, width=0.5\textwidth]
 {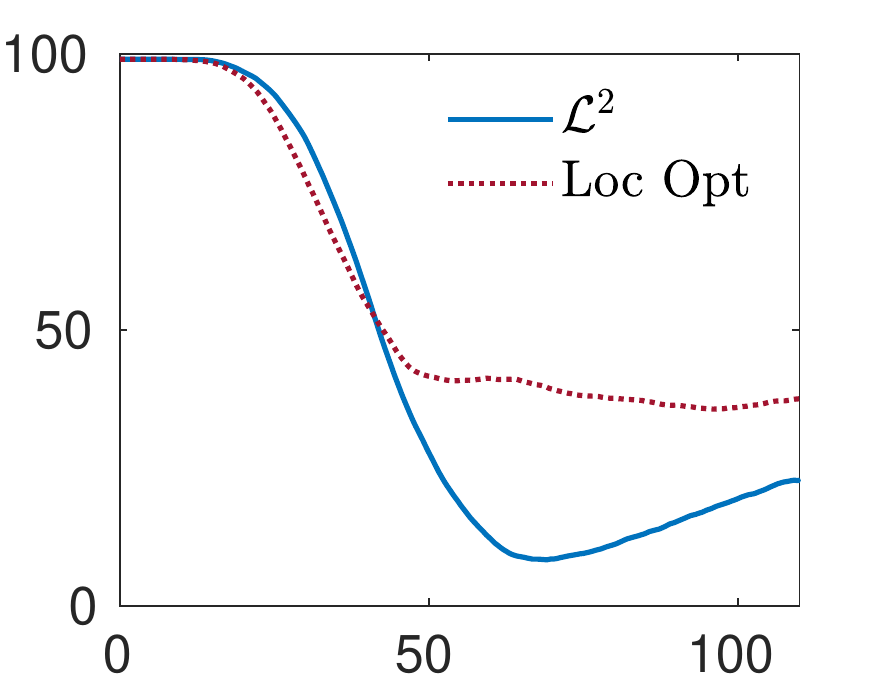}}
 \subfloat[Parameter $2$]{
 \includegraphics
 [height=0.22\textheight, width=0.5\textwidth]
 {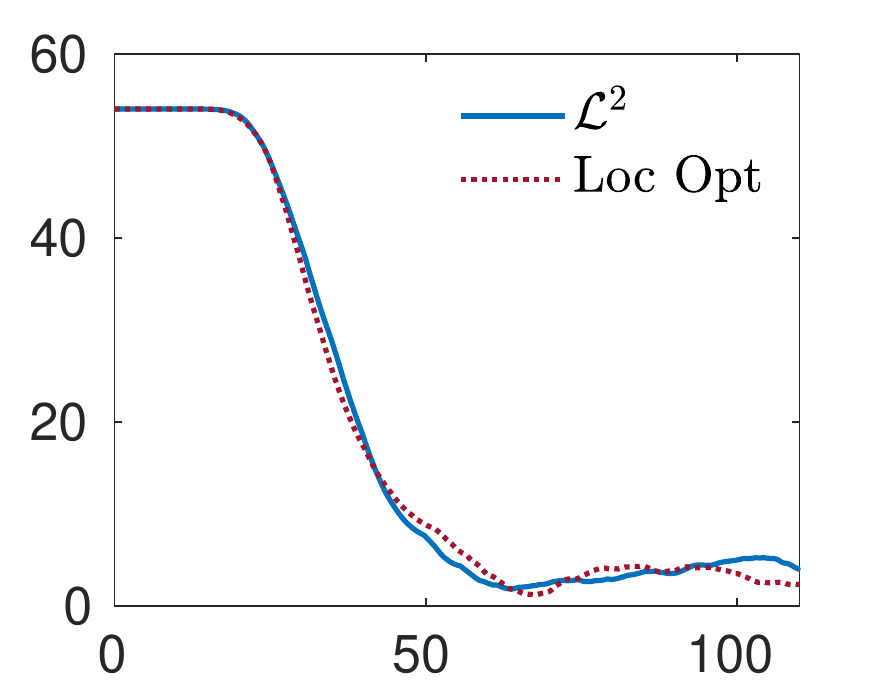}}
 \caption{Avg. parameter estimation errors with time - with
 directed consensus.}
 \label{fig:avgparerror_dirconsensus_rgb}
\end{figure}

\section{Conclusion}
\label{sec:conclusions}
We have looked at an alternative framework for defining the
coverage problem. The sensing quality of each agent was quantified
as the agent sensing function, and an aggregate sensing function
was formed. The coverage problem was then defined as the
minimization of some distance between the aggregate function and
the density function. We showed that the locational optimization
problem can be viewed as a special case of this framework using the
K-L divergence as the distance measure. We also looked at the
$\mathcal{L}^2$ distance as a metric for coverage, and compared
the performance of $\mathcal{L}^2$ coverage with the locational
optimization based coverage.

%


\bibliographystyle{spbasic}
\bibliography{mybibfile.bib}   

\begin{thebibliography}{22}
\providecommand{\natexlab}[1]{#1}
\providecommand{\url}[1]{{#1}}
\providecommand{\urlprefix}{URL }
\expandafter\ifx\csname urlstyle\endcsname\relax
  \providecommand{\doi}[1]{DOI~\discretionary{}{}{}#1}\else
  \providecommand{\doi}{DOI~\discretionary{}{}{}\begingroup
  \urlstyle{rm}\Url}\fi
\providecommand{\eprint}[2][]{\url{#2}}

\bibitem[{Abdul~Razak et~al.(2018)Abdul~Razak, Srikant, and Chung}]{rihab2018a}
Abdul~Razak R, Srikant S, Chung H (2018) Decentralized and adaptive control of
  multiple nonholonomic robots for sensing coverage. International Journal of
  Robust and Nonlinear Control 28(6):2636--2650

\bibitem[{{Bopardikar} et~al.(2018){Bopardikar}, {Mehta}, and
  {Hauenstein}}]{bopardikar2018arXiv}
{Bopardikar} SD, {Mehta} D, {Hauenstein} JD (2018) {Optimal Configurations in
  Coverage Control with Polynomial Costs}. ArXiv e-prints \eprint{1801.10285}

\bibitem[{Bullo et~al.(2009)Bullo, Cort\'es, and
  Mart{\'\i}nez}]{DistCtrlRobotNetw}
Bullo F, Cort\'es J, Mart{\'\i}nez S (2009) Distributed Control of Robotic
  Networks. Applied Mathematics Series, Princeton University Press,
  electronically available at http://coordinationbook.info

\bibitem[{Cortes and Bullo(2005)}]{cortes2005geometric}
Cortes J, Bullo F (2005) Coordination and geometric optimization via
  distributed dynamical systems. SIAM Journal on Control and Optimization
  44(5):1543--1574

\bibitem[{Cortes et~al.(2004)Cortes, Martinez, Karatas, and
  Bullo}]{cortes2004coverage}
Cortes J, Martinez S, Karatas T, Bullo F (2004) Coverage control for mobile
  sensing networks. IEEE Trans on Automatic Control 20(2):243--255

\bibitem[{Cortes et~al.(May, 2002)Cortes, Martinez, Karatas, and
  Bullo}]{Cortes2002}
Cortes J, Martinez S, Karatas T, Bullo F (May, 2002) {Coverage control for
  mobile sensing networks}. In: Proc. {IEEE} Int. Conf. Robot. Autom., pp
  1327--1332

\bibitem[{Flanders(1973)}]{flanders73}
Flanders H (1973) {Differentiation Under the Integral Sign}. Amer Math Monthly
  80(6):615--627

\bibitem[{Guruprasad and Ghose(2013)}]{guruprasad2013}
Guruprasad KR, Ghose D (2013) {Heterogeneous locational optimisation using a
  generalised Voronoi partition}. Int J Control 86(6):977--993

\bibitem[{Hexsel et~al.(2011)Hexsel, Chakraborty, and Sycara}]{hexsel2011}
Hexsel B, Chakraborty N, Sycara K (2011) Coverage control for mobile
  anisotropic sensor networks. In: 2011 IEEE Int. Conf. Robot. and Autom., pp
  2878--2885

\bibitem[{Jadbabaie and Lin(2003)}]{jadbabaie2003coordination}
Jadbabaie A, Lin J (2003) {Coordination of groups of mobile autonomous agents
  using nearest neighbor rules}. IEEE Trans on Automatic Control
  48(6):988--1001

\bibitem[{Khalil(2002)}]{khalil2002nonlinear}
Khalil H (2002) Nonlinear Systems. Pearson Education, Prentice Hall

\bibitem[{Kullback and Leibler(1951)}]{kullback1951}
Kullback S, Leibler RA (1951) On information and sufficiency. Ann Math Statist
  22(1):79--86

\bibitem[{Luna et~al.(2013)Luna, Fierro, Abdallah, and Wood}]{Luna2013}
Luna JM, Fierro R, Abdallah CT, Wood J (2013) {An Adaptive Coverage Control for
  Deployment of Nonholonomic Mobile Sensor Networks Over Time-Varying Sensory
  Functions}. Asian J Control 15(4):988--1000

\bibitem[{Murray(2007)}]{murray2007recent}
Murray RM (2007) Recent research in cooperative control of multivehicle
  systems. J Dynam Syst Measur and Control 129(5):571--583

\bibitem[{Olfati-Saber et~al.(2007)Olfati-Saber, Fax, and
  Murray}]{olfati2007consensus}
Olfati-Saber R, Fax A, Murray RM (2007) Consensus and cooperation in networked
  multi-agent systems. Proceedings of the IEEE 95(1):215--233

\bibitem[{Press(2007)}]{numericalrecipes2007}
Press W (2007) Numerical Recipes 3rd Edition: The Art of Scientific Computing.
  Cambridge University Press

\bibitem[{{Razak} et~al.(2018){Razak}, {Sukumar}, and {Chung}}]{rihab2018b}
{Razak} RA, {Sukumar} S, {Chung} H (2018) Distributed coverage control of
  mobile sensors: Generalized approach using distance functions. In: 2018 IEEE
  Conference on Decision and Control (CDC), pp 3323--3328

\bibitem[{Schwager et~al.(2009)Schwager, Rus, and Slotine}]{Schwager2009}
Schwager M, Rus D, Slotine JJ (2009) Decentralized, adaptive coverage control
  for networked robots. Int J Rob Res 28(3):357--375

\bibitem[{Schwager et~al.(Apr. 10-14, 2007)Schwager, Slotine, and
  Rus}]{SchwagerEtalICRA07}
Schwager M, Slotine JE, Rus D (Apr. 10-14, 2007) Decentralized, adaptive
  control for coverage with networked robots. In: Proc. {IEEE} Int. Conf.
  Robot. Autom., pp 3289--3294

\bibitem[{Song et~al.(2011)Song, Feng, Fan, and Wang}]{Song2011}
Song C, Feng G, Fan Y, Wang Y (2011) {Decentralized adaptive awareness coverage
  control for multi-agent networks}. Automatica 47(12):2749--2756

\bibitem[{Song et~al.(2013)Song, Liu, Feng, Wang, and Gao}]{Song2013}
Song C, Liu L, Feng G, Wang Y, Gao Q (2013) {Persistent awareness coverage
  control for mobile sensor networks}. Automatica 49(6):1867--1873

\bibitem[{Tanner et~al.(2007)Tanner, Jadbabaie, and
  Pappas}]{tanner2007flocking}
Tanner H, Jadbabaie A, Pappas G (2007) Flocking in fixed and switching
  networks. IEEE Trans on Automatic Control 52(5):863--868

\end{thebibliography}


\end{document}